\documentclass[11pt,a4paper]{article}
\usepackage[T1]{fontenc}
\usepackage[utf8]{inputenc}
\usepackage[english]{babel}
\usepackage{lmodern}
\usepackage[margin=1in,headheight=14pt]{geometry}
\usepackage{amsmath,amssymb,amsthm,mathtools}
\usepackage{microtype}
\usepackage{booktabs,tabularx,array}
\usepackage[shortlabels]{enumitem}
\usepackage{aliascnt}
\usepackage{needspace}
\usepackage{fancyhdr}
\usepackage{xurl}
\usepackage[hidelinks,bookmarksnumbered=true]{hyperref}
\usepackage[nameinlink,noabbrev]{cleveref}
\numberwithin{equation}{section}

\theoremstyle{plain}
\newtheorem{theorem}{Theorem}[section]
\newaliascnt{lemma}{theorem}
\newtheorem{lemma}[lemma]{Lemma}
\aliascntresetthe{lemma}
\newaliascnt{proposition}{theorem}
\newtheorem{proposition}[proposition]{Proposition}
\aliascntresetthe{proposition}
\newaliascnt{corollary}{theorem}
\newtheorem{corollary}[corollary]{Corollary}
\aliascntresetthe{corollary}
\theoremstyle{definition}
\newaliascnt{definition}{theorem}
\newtheorem{definition}[definition]{Definition}
\aliascntresetthe{definition}
\theoremstyle{remark}
\newaliascnt{remark}{theorem}
\newtheorem{remark}[remark]{Remark}
\aliascntresetthe{remark}
\crefname{theorem}{Theorem}{Theorems}
\crefname{lemma}{Lemma}{Lemmas}
\crefname{proposition}{Proposition}{Propositions}
\crefname{corollary}{Corollary}{Corollaries}
\crefname{definition}{Definition}{Definitions}
\crefname{remark}{Remark}{Remarks}
\crefname{section}{Section}{Sections}
\crefname{appendix}{Appendix}{Appendices}

\hypersetup{
  pdftitle={Linear Certificates for Membership Comparability, Quadratic Barriers for Selectors},
  pdfauthor={Sebastian Ben Daniel},
  pdfsubject={Public preprint: membership comparability, advice complexity, and black-box barriers},
  pdfkeywords={Membership comparability, P-selective sets, advice complexity, short certificates, average-case complexity, black-box lower bounds, relativization}}

\newcommand{\Psel}{\mathrm{P\text{-}sel}}
\newcommand{\Pclass}{\mathrm P}
\newcommand{\BPP}{\mathrm{BPP}}
\newcommand{\E}{\mathbb E}
\newcommand{\F}{\mathbb F}
\newcommand{\bits}{\{0,1\}}

\newcommand{\ind}{\mathbf 1}
\newcommand{\poly}{\operatorname{poly}}

\newcommand{\unk}{\mathord{?}}

\newcommand{\NP}{\mathrm{NP}}
\newcommand{\PSPACE}{\mathrm{PSPACE}}
\newcommand{\PselW}{\mathrm{P}^{W}\text{-}\mathrm{sel}}
\newcommand{\AR}{\operatorname{AR}}
\newcommand{\Core}{\mathcal T}

\newcommand{\Mc}{\mathrm{2\text{-}mc}}

\title{Linear Certificates for Membership Comparability,\\
Quadratic Barriers for Selectors}
\author{Sebastian Ben Daniel}
\date{September 25, 2026}

\begin{document}
\maketitle
\begin{abstract}
Selectors and comparators supply only partial information about membership: a
selector names a member of any pair that meets the language, while a binary
membership comparator merely excludes one of the four membership vectors of a
pair. We ask how much nonuniform advice turns such information into exact
recognition.

Our main result extends the optimal nondeterministic advice bound for P-selective
sets to every binary membership-comparable language:
$\Mc\subseteq\NP/(3n+5)\cap\mathrm{coNP}/(3n+5)$, with common fixed advice and
certificates of at most $5n+12$ bits. The class is strictly larger; some
$\Mc$ sets are not truth-table reducible to any P-selective set. The proof
replaces the tournament king by an independent two-step cover of true signed
literals, together with a short-forcing-or-exact-majority dichotomy, and it
relativizes. Via an advice-preserving isolation transfer, a deterministic
polynomial-time algorithm for promise Unique-Circuit-SAT gives $\Mc\subseteq\Pclass/O(n)$.

For selectors we determine tight orders of ordinary advice: $\Theta(n)$ for
errorless average-case computation and $\Theta(n^2)$ for worst-case bounded-error
computation, the latter independent of the interpreter's coin bound. One oracle
realizes both orders on a single language and separates ordinary from
coin-dependent advice. The quadratic and linear lower bounds hold for
tournament-query procedures and relativized languages, not unconditionally for
unrelativized P-selective sets.
\end{abstract}
\clearpage
\tableofcontents
\clearpage

\section{Introduction}
\label{m:intro}\label{sec:introduction}

A polynomial-time selector compares two strings and returns a member of a language
whenever at least one is a member. A binary membership comparator provides still
less information: it excludes one of the four possible membership vectors of the
two strings, without necessarily determining either bit. How much nonuniform
information turns such partial information into exact recognition? We study this
question across nondeterministic verification, errorless average-case computation,
and worst-case bounded-error computation. Ordinary advice is fixed once per input
length, independently of both the input and the interpreter's coins or witnesses.

\subsection{Results and scope}

Our first result extends fixed linear nondeterministic advice from P-selective
sets to \emph{unrestricted binary membership comparability}:
\begin{equation}
 \Mc\subseteq\NP/(3n+5)\ \cap\ \mathrm{coNP}/(3n+5).
 \label{m:intro-mc}
\end{equation}
The two interpreters share the same advice. Certificates have at most $5n+12$
bits and verification uses at most six comparator evaluations, all on distinct
equal-length inputs. The proof replaces a tournament king by an independent
two-step cover of true signed literals. Either two advised literals and short
forcing paths suffice, or three comparator rows predict every membership bit by
exact majority. We give the complete proof in \cref{m:comparators}. The conclusion
also holds for lengthwise comparators and relative to every oracle. Combined with
our advice-preserving isolation theorem, it yields the one-line conditional
extension $\mathsf U\Rightarrow\Mc\subseteq\Pclass/O(n)$.

For selectors, we determine tight \emph{orders} of ordinary advice in two further
regimes. Every P-selective language has an exact deterministic decider with
$3n+O(1)$ advice and polynomial uniform mean time, and a deterministic
polynomial-time errorless heuristic with $n+O_c(\log n)$ advice and uniform
abstention at most $n^{-c}$. Transcript freezing forces $n-O_q(\log n)$ advice
for polynomial-query errorless heuristics even at coverage $1/2$, including
randomized heuristics with any finite coin bound. Truncation gives the matching
linear lower order for exact polynomial mean time.

In the worst case, a hidden-core distribution forces
\begin{equation}
 (n-\log q)^2/4-O(n\log n)
 \label{m:intro-quadratic}
\end{equation}
ordinary advice bits for bounded-error $q$-query deciders, independently of the
finite coin bound. For polynomially many queries this is quadratic, matching
Ko's deterministic polynomial-time upper order \cite{Ko83}. One oracle-relative
P-selective language realizes both the linear average-case and quadratic
worst-case orders, and has linear coin-dependent advice with logarithmically
many coins. The new comparator theorem supplies fixed linear nondeterministic
advice throughout the larger class; it does not assert deterministic linear
advice or extend the selector-specific density bounds to all comparators.

\begin{table}[t]
\caption{Advice bounds, interpreter resources, and scope. Unless specified,
$q$ is a fixed polynomial and $c>0$ is fixed. BB denotes the black-box model;
REL denotes a relativized result. Unrelativized upper bounds listed here also
relativize, except for the explicitly conditional row.}
\label{m:table}
\begin{tabularx}{\textwidth}{@{}>{\raggedright\arraybackslash}p{.25\textwidth}
 >{\raggedright\arraybackslash}X
 >{\raggedright\arraybackslash}p{.20\textwidth}@{}}
\toprule
Class and task & Advice and conclusion & Scope\\
\midrule
$\Mc$: NP and coNP verification
& Common $3n+5$ bits; witnesses $\le5n+12$ bits; at most six comparator calls.
  \Cref{m:mc-main}.
& Unrelativized upper; BB proof; REL.\\
\addlinespace
$\Psel$: errorless uniform heuristics
& $n-O_q(\log n)$ necessary for coverage $1/2$; $n+O_c(\log n)$ sufficient
  for abstention $n^{-c}$. Thus $\Theta(n)$. \Cref{m:heuristic,m:near-n}.
& Unrelativized upper; BB/REL lower.\\
\addlinespace
$\Psel$: exact deterministic uniform mean time
& $3n+O(1)$ sufficient; $\Omega(n)$ necessary for polynomial mean time.
  Thus $\Theta(n)$. \Cref{m:det,m:joint}.
& Unrelativized upper; BB/REL lower.\\
\addlinespace
$\Psel$: exact randomized mean time
& $2n+O(1)$ sufficient under every exactly polynomial-time samplable input
  distribution. \Cref{m:random}.
& Unrelativized upper; REL.\\
\addlinespace
$\Psel$: worst-case polynomial time, any polynomial coin bound
& $n^2/4-O_q(n\log n)$ necessary; Ko's deterministic $n^2+n$ bits sufficient.
  Thus $\Theta(n^2)$. \Cref{m:quadratic}; \cite{Ko83}.
& Unrelativized upper; BB/REL lower.\\
\addlinespace
Tournaments: query complexity only
& Lower coefficient $1/4$, upper coefficient $1/2$ with $n+2$ queries.
  No polynomial-time indexing guarantee. \Cref{m:half}.
& BB; unrestricted internal computation.\\
\addlinespace
One language; ordinary versus coin-dependent advice
& One $W,L$ realizes both tight orders; coin-dependent $O(n)$ advice and
  $\log n+O(1)$ coins suffice. \Cref{m:joint,m:models}.
& REL; the same $\mathrm P^W$-selective language.\\
\addlinespace
$\Mc$: code access under $\mathsf U$
& $\Mc\subseteq\Pclass/O(n)$; an unrelativized counterexample refutes
  $\mathsf U$. \Cref{m:mc-U,m:generic}.
& Unrelativized, conditional.\\
\bottomrule
\end{tabularx}
\end{table}

\paragraph*{Scope of the lower bounds and constants.}
The average-case and worst-case lower bounds concern arbitrary tournament-query
procedures and a single relativized language; they are not unconditional
quadratic lower bounds for unrelativized P-selective sets or $\Mc$.
For polynomial-time decoding, our comparison remains against \emph{Ko's
$n^2+n$-bit upper bound, with coefficient $1$}. A separate universal-family
argument yields $n^2/2+O(n\log n)$ bits only when arbitrary finite,
oracle-independent computation is free. Efficient indexing of that family is
not proved. The constant gap is therefore $1/4$ versus $1/2$ in the query model,
and $1/4$ versus $1$ for polynomial-time interpreters.

\subsection{Proof overview and context}

\paragraph*{Comparators: missing edges become forcing certificates.}
Every excluded vector defines a valid two-literal clause and two implications
between signed membership assertions. The true literals induce an oriented
graph, not necessarily a tournament. A Chv\'atal--Lov\'asz independent outward
two-step cover \cite{CL74} replaces the single tournament king. Declare a literal
short-forced when its negation implies it within four steps. Among the covering
literals that are not short-forced, no true literal can have two out-neighbors:
nonadjacency inside the cover would close a four-step forcing path. If at most
two remain, advise them; otherwise three of them give row predictions with at
most one error on each input. The verifier checks only local implications, not
the cover or the unproved assertion that a guessed literal is true.

\paragraph*{Selector upper bounds: a certificate-density profile.}
A positive minimum-outdegree anchor and a negative minimum-indegree anchor give
sound two-step certificates. Their density $\rho(x)$ satisfies an atom bound
$\rho(x)\ge\nu(x)$ and a lower-tail bound
$\nu\{x:\rho(x)\le t\}\le4t$. Reciprocal integration bounds randomized mean
time. One common shift of small-bias sets replaces all search randomness for
the $3n$-bit exact decider. For a bounded-time heuristic, one shift regenerates
a polynomial-size sample, and one cut in a Hamiltonian path encodes all its labels.

\paragraph*{Two lower bounds and one oracle.}
Freezing maintains a transitive order with an uncommitted middle block. Each
advice that answers there can be made permanently wrong while fixing only its
input and query endpoints. The hidden core instead hides one point per bucket
and independent labels inside that sparse set. Few discoveries leave fresh
label-and-edge choices; many discoveries are unlikely. The same hard distribution
works for every deterministic procedure, so Yao's principle plus $O(n)$-fold
amplification lifts the bound without counting random tapes. Interleaving the
two requirements at sparse lengths yields the joint oracle. Padded membership
records are hidden from each diagonalized machine but accessible to the selector
on sufficiently separated input lengths.

\paragraph*{Code access and prior work.}
Under deterministic promise Unique-Circuit-SAT, the comparator's program can be
compiled into an ordinary circuit with only linearly many witness inputs.
Linear-seed isolation and one short expander walk then give deterministic linear
advice; the general transfer preserves advice $a(n)$ and witnesses $w(n)$ within
$O(a(n)+w(n)+n)$ bits. This uses code unavailable from tournament answers alone.

Ko's quadratic bound \cite{Ko83}, the optimal nondeterministic $(n+1)$-bit
selector bound \cite{HT96}, and its linear-witness refinement \cite{HNP98} set the
scale. Constant-arity comparability was known to imply $\Pclass/\mathrm{poly}$
\cite{ABG03}, and exponential-time interpretation already permits linear advice
\cite{Nic97}; our binary theorem adds polynomial-time, linear-witness verification
with common fixed advice. Binary comparability is genuinely broader than
P-selectivity \cite{BFP}. Earlier work gives linear coin-dependent bounded-error
advice for $\Mc$, with an errorless selector version \cite{BD12}; fixed-advice
nondeterminism is a different conclusion. Thakur's qualitative oracle barrier
\cite{Thakur03} predates our quantitative lower bounds. Our proofs also use
classical small-bias constructions, isolation, and expanders
\cite{AGHP92,VV86,RVW02}; the randomized lift is Yao's principle, not a new
minimax principle \cite{Yao77}. Associative selectors give a different known
sufficient condition for deterministic linear advice \cite{HHN04}.

The main text proves the binary theorem, freezing, and randomized lifting in
full. Appendices give the other proofs, exact parameter checks, and constructions.

\section{Preliminaries: selectors, comparators, and advice}
\label{m:prelim}

All logarithms are base two except $\ln$. Ordinary advice $a_n$ is fixed for all
$n$-bit inputs and independent of coins and witnesses; it need not be computable.
For NP and coNP advice, correctness is required only for the designated sequence.
The slash in $\NP/a(n)\cap\mathrm{coNP}/a(n)$ binds separately to each class;
common advice for both verifiers is an additional assertion.
An errorless heuristic outputs only $\chi_L(x)$ or $\unk$ on every input and tape.
Uniform coverage averages conclusive probability over $x\sim U_n$ and coins.
Exact polynomial mean time is a polynomial expectation under $U_n$, not a
pointwise time bound; randomized exact deciders terminate almost surely.

A P-selector $f$ satisfies $f(x,y)\in\{x,y\}$ and returns a member whenever
$\{x,y\}\cap L\ne\varnothing$. Order its arguments first to make it commutative.
For distinct $n$-bit strings write $x\to y$ when $f(x,y)=x$; every vertex in
$P=L\cap\bits^n$ beats every vertex in $\overline P$.
A binary membership comparator is a polynomial-time function
$g(x,y)\in\bits^2$ unequal to $(\chi_L(x),\chi_L(y))$; its class is $\Mc$.
A selector is a comparator that always excludes a mixed vector: if it returns
$x$, exclude $(0,1)$, and otherwise exclude $(1,0)$. Comparators may also exclude
$(0,0)$ or $(1,1)$. Our binary theorem needs correctness only on distinct inputs
of equal length. Normalize each pair lexicographically, and denote the returned
bits associated to the respective inputs by $\beta_x(x,y),\beta_y(x,y)$.
This costs one comparator call and ignores self-comparisons.

\subsection{Tournament anchors for the selector upper bounds}

Weights restricted to an induced subtournament are not renormalized.
\begin{lemma}[Two-step witness tail]\label{m:weighted}
In a finite tournament with nonnegative weights $\nu$, let $p$ minimize weighted
outdegree. For $p\to x$, put $I_x=\{z:x\to z\to p\}$ and
$J_x=\{z:p\to z\to x\}$. Then, for $t\ge0$,
\begin{equation}
 \nu(I_x)\ge\nu(J_x)+\nu(x),\qquad
 \nu\{x:p\to x,\ \nu(I_x)\le t\}\le2t.
 \label{m:eq-weighted}
\end{equation}
\end{lemma}
The identity $d^+(x)-d^+(p)=\nu(I_x)-\nu(J_x)-\nu(x)$ proves the first bound.
Summing over $A_t=\{x:p\to x,\nu(I_x)\le t\}$ counts each pair once and gives
$t\nu(A_t)\ge(\nu(A_t)^2+\sum_{x\in A_t}\nu(x)^2)/2$.
For unit weights, every vertex reaches the minimum-outdegree vertex in at most
two steps. See \cref{sec:preliminaries} for the full weighted argument.

For a full-support distribution $\nu$ on $X=\bits^n$, advise a positive
minimum-outdegree anchor $p$ and a negative minimum-indegree anchor $q_-$, with
empty-class flags: $2n+O(1)$ bits. Accept directly if $x=p$ or $x\to p$;
reject if $x=q_-$ or $q_-\to x$. For unresolved $x\in\mathcal H$, put
\[
 C_+(x)=\{z:x\to z\to p\},\quad C_-(x)=\{z:q_-\to z\to x\},\quad
 \rho(x)=\nu(C_+(x)\cup C_-(x)).
\]
The two path types certify opposite membership bits soundly. Applying
\cref{m:weighted} in the positive tournament and the reversed negative tournament gives
\begin{equation}
 \rho(x)\ge\nu(x)>0\ (x\in\mathcal H),\qquad
 \nu\{x\in\mathcal H:\rho(x)\le t\}\le4t.
 \label{m:eq-density}
\end{equation}
No labels for intermediate vertices are advised: $z\to p$ proves positivity and
$q_-\to z$ proves negativity. The atom bound guarantees eventual completeness;
the tail bound controls average cost, not a pointwise inverse-polynomial density.

\section{Binary comparators: fixed linear advice and short certificates}
\label{m:comparators}

\begin{theorem}[Common linear advice for binary comparability]\label{m:mc-main}
For every $L\in\Mc$ there are a deterministic polynomial-time verifier $V$ and
advice $a_n$ of length $3n+5$ such that, for every $x\in\bits^n$ and $b\in\bits$,
\begin{equation}
 b=\chi_L(x)\quad\Longleftrightarrow\quad
 \exists w\in\bits^{\le5n+12}\ V(x,b,a_n,w)=1.
 \label{m:eq-mc-verifier}
\end{equation}
Verification uses at most six comparator calls, on distinct $n$-bit strings.
Hence $\Mc\subseteq\NP/(3n+5)\cap\mathrm{coNP}/(3n+5)$, with the same advice.
Only distinct equal-length comparator correctness is required.
\end{theorem}

\subsection{Signed implications and independent two-step covers}

Fix $n$, put $c(x)=\chi_L(x)$, and let $\langle x,b\rangle$ denote $c(x)=b$.
Its complement is $\overline{\langle x,b\rangle}=\langle x,1-b\rangle$; write
$\ell_x=\langle x,c(x)\rangle$ for the true literal on $x$.
For each distinct pair, the excluded vector gives the valid clause
\[
 \langle x,1-\beta_x(x,y)\rangle\ \vee\ \langle y,1-\beta_y(x,y)\rangle.
\]
Let $G$ contain its two implication edges. Equivalently, for $x\ne y$,
\begin{equation}
 \langle x,a\rangle\to\langle y,d\rangle\text{ in }G
 \quad\Longleftrightarrow\quad
 \beta_x(x,y)=a\text{ and }\beta_y(x,y)=1-d.
 \label{m:eq-mc-edge}
\end{equation}
Edges on the same underlying input are disallowed. Each edge is checked with
one comparator call, and a true source implies a true destination: otherwise
the comparator would have excluded the actual vector. Every walk is therefore
a sound implication. Write $u\leadsto_F^{\le r}v$ for a walk of at most $r$
edges in $F$; length zero means $u=v$.

Let $D=G[\{\ell_x:x\in X\}]$ be the true-literal graph. It is oriented but may
have missing edges. For each pair, exclusion of $(c(x),1-c(y))$ gives
$\ell_x\to\ell_y$, exclusion of $(1-c(x),c(y))$ gives the reverse, and exclusion
of $(1-c(x),1-c(y))$ gives neither. In the last case $G$ contains
\begin{equation}
 \overline{\ell_x}\to\ell_y,\qquad \overline{\ell_y}\to\ell_x.
 \label{m:eq-mc-nonedge}
\end{equation}
Thus a missing true-literal edge supplies a disjunction, not an absent constraint.
The verifier evaluates $G$ only; it never tests membership in $D$.

\begin{lemma}[Independent outward two-step cover \cite{CL74}]\label{m:mc-cover}
Every finite loopless digraph $F$ has an independent set $Q$ such that every
vertex is reached from some $q\in Q$ within two steps.
\end{lemma}
\begin{proof}
Induct on $|V(F)|$. The empty graph is immediate. Choose $v$ and apply induction
to $F'=F-(\{v\}\cup N_F^+(v))$, obtaining $Q'$. If some $q\in Q'$ has $q\to v$,
then $Q'$ covers $F'$ by induction, $v$ in one step, and its out-neighbors in two.
Otherwise use $Q'\cup\{v\}$. There are no edges from $Q'$ to $v$ by this case
assumption and none from $v$ to $Q'$ by the definition of $F'$. Independence and
coverage follow. This is the usual quasi-kernel theorem with the graph reversed.
\end{proof}

Choose this cover $Q$ in $D$. Call a literal $t$ \emph{short-forced} if
$\overline t\leadsto_G^{\le4}t$. Such a literal is true: if false, its complement
would be true and the walk would force $t$ true. Put
$H=\{h\in Q:h\text{ is not short-forced}\}$.
No algorithm for constructing $Q$ or testing non-forcing is required.

\begin{lemma}[At most one out-neighbor among unforced covering literals]
\label{m:mc-one}
Every true literal $\ell_x$ has at most one out-neighbor in $H$.
\end{lemma}
\begin{proof}
Suppose $\ell_x\to h_1,h_2$ for distinct $h_1,h_2\in H$. Choose $q\in Q$ with
$q\leadsto_D^{\le2}\ell_x$. Some $h_i$ differs from $q$. Independence of $Q$
and \cref{m:eq-mc-nonedge} give $\overline{h_i}\to q$ in $G$, so
\begin{equation}
 \overline{h_i}\ \to\ q\ \leadsto_D^{\le2}\ \ell_x\ \to\ h_i
 \label{m:eq-mc-force}
\end{equation}
is a forcing walk of at most four edges, contradicting $h_i\in H$.
\end{proof}

\subsection{Short certificates or an exact three-row majority}

\begin{proof}[Proof of \cref{m:mc-main}]
The advice has a two-bit header $j\in\{0,1,2,3\}$ and three signed-literal slots,
each of $n+1$ bits. Its length is $2+3(n+1)=3n+5$.

\paragraph*{Certificate mode: $|H|\le2$.}
Set $j=|H|$ and advise all of $H$ in the first $j$ slots. For a target literal
$t=\langle x,b\rangle$, a witness supplies either
\[
 h\leadsto_G^{\le2}t\quad(h\text{ an advised literal}),
 \qquad\text{or}\qquad
 \overline q\leadsto_G^{\le4}q\ \text{ and }\ q\leadsto_G^{\le2}t.
\]
Every advised premise is true; every accepted forcing walk certifies its premise.
Soundness of the remaining walk therefore proves the target true on every
accepting branch. Conversely, if $t$ is true, some $q\in Q$ reaches it within
two steps in $D$. If $q\in H$, use its advised premise; otherwise $q$ is
short-forced and the second certificate exists. Guessed literals need not be
asserted to lie in $Q$ or $D$: their truth follows from checked implications.
The header $j=0$ means no advised premise is needed, not that the slice is empty.

\paragraph*{Majority mode: $|H|\ge3$.}
Set $j=3$ and advise three distinct literals $h_i=\ell_{z_i}\in H$. On an input
$x=z_i$, use its advised label. Otherwise form
$p_i(x)=1-\beta_x(z_i,x)$ for $i=1,2,3$. If $p_i(x)\ne c(x)$, then
$\beta_x(z_i,x)=c(x)$, so comparability forces
$\beta_{z_i}(z_i,x)=1-c(z_i)$. By \cref{m:eq-mc-edge}, $\ell_x\to h_i$ in $D$.
\Cref{m:mc-one} permits at most one such error. Hence
\begin{equation}
 c(x)=\operatorname{MAJ}\bigl(1-\beta_x(z_1,x),1-\beta_x(z_2,x),
                         1-\beta_x(z_3,x)\bigr).
 \label{m:eq-mc-majority}
\end{equation}
This mode is deterministic: accept exactly when this answer equals $b$, using
an empty witness. Reject advice with repeated majority inputs.

\paragraph*{Encoding and verification.}
A certificate-mode witness has a seven-bit header: one type bit, one advised-slot
bit, three bits for forcing length $r\in\{1,2,3,4\}$, and two for final length
$s\in\{0,1,2\}$. Five $(n+1)$-bit slots suffice: $q$, up to three internal
forcing literals, and the possible internal final literal. Endpoints are
reconstructed from $q$, the advice, and $t$; unused fields are ignored.
For an advised premise only the slot index and possible final internal literal
are needed. The length is $7+5(n+1)=5n+12$. Reject malformed advice or witness lengths and invalid used fields; a
zero-edge final walk requires identical endpoints. Verification tests at most
$r+s\le6$ edges, using \cref{m:eq-mc-edge} only; majority mode makes at most three
calls. Its time is $O(T_g(n)+n)$ for comparator time $T_g$. This also handles
$n=0$, when certificate mode advises the single true literal and uses a zero-edge
walk. The designated advice works simultaneously for both $b=0$ and $b=1$,
proving \cref{m:eq-mc-verifier} and the two containments.
\end{proof}

\begin{corollary}[Relativization]\label{m:mc-rel}
For every oracle $A$, $\Mc^A\subseteq\NP^A/(3n+5)\cap\mathrm{coNP}^A/(3n+5)$,
with the same common-advice, witness, and six-call bounds.
\end{corollary}
\begin{corollary}[Promise-unique transfer]\label{m:mc-U}\label{m:conditional}
Under $\mathsf U$ from \cref{m:conditional-section}, $\Mc\subseteq\Pclass/O(n)$,
by \cref{m:generic} with $a(n)=3n+5$ and $w(n)=5n+12$.
\end{corollary}

The first corollary substitutes the oracle comparator in the identical verifier;
the proof is entirely lengthwise. For the second, pad certificates to their fixed
maximum length before applying the transfer. Correctness is required only under
the designated advice, not on arbitrary advice--input pairs. We do not claim
that the universal verifier defines an $\NP\cap\mathrm{coNP}$ language for all
advice strings. The linear order is already necessary on the P-selective subclass, since
$\Psel\not\subseteq\NP/n$ and $\Psel\not\subseteq\mathrm{coNP}/n$ \cite{HT96,HNP98}; neither coefficient
$3$ nor the witness and call constants are proved optimal. The argument is binary:
for higher arity, missing adjacency no longer supplies the one-premise implication
used in \cref{m:eq-mc-force}.

\section{Selector upper bounds: linear advice on average}
\label{m:upper}

These results concern every $L\in\Psel$ and relativize. They use only the
selector's evaluations. Full proofs and constructions are in
\cref{sec:upper,app:technical}; no extension of the density argument to all
binary comparators is assumed.

\begin{theorem}[Exact deterministic polynomial mean time]\label{m:det}
Every $L\in\Psel$ has a deterministic exact decider $D$ with ordinary advice
$|a_n|=3n+O(1)$ and a polynomial $Q$ such that
\begin{equation}
 D(x,a_n)=\chi_L(x)\quad(x\in\bits^n),\qquad
 \E_{x\sim U_n}T_D(x,a_n)\le Q(n).
 \label{m:eq-det}
\end{equation}
Its worst-case time is $2^n\poly(n)$. The mean charged number of certificate
candidates, including setup for every reached stage, is $O(n^4)$. One advice
sequence supports every fixed inverse-polynomial timeout guarantee.
\end{theorem}

\begin{theorem}[Exact randomized polynomial mean time]\label{m:random}
For every $L\in\Psel$ and exactly polynomial-time samplable ensemble $\mu_n$,
an errorless interpreter with $2n+O(1)$ ordinary advice terminates almost surely
on every input and has polynomial mean time over $x\sim\mu_n$ and its coins.
Its mean certificate-sample count is $O(n)$. Truncation after $M$ samples has
average abstention at most $8/(M+1)$, or $4/(M+1)$ for direct uniform sampling.
\end{theorem}

\paragraph*{Independent sampling and a full-support mixture.}
After the direct tests, independently sample intermediate vertices from $\nu$
until a certificate appears. On unresolved $x$ the mean sample count is
$1/\rho(x)$. Integrating \cref{m:eq-density} gives
\begin{equation}
 \E_{x\sim\nu}\frac{\ind_{\mathcal H}(x)}{\rho(x)}
 \le1+4\ln(1/\beta)\quad\text{when }\nu(x)\ge\beta>0.
 \label{m:eq-reciprocal}
\end{equation}
A second integration bounds the average probability of missing after $M$ trials
by $4/(M+1)$. For arbitrary $\mu_n$ choose anchors and samples using
$\nu_n=(\mu_n+U_n)/2$: then $\nu_n(x)\ge2^{-n-1}$ and $\mu_n\le2\nu_n$.
This preserves eventual answers even outside the support of $\mu_n$, with
polynomial mean under $\mu_n$. The sampler may use polynomially many coins.
The expectation includes inputs, so this is not a $\BPP/O(n)$ containment.

\paragraph*{One shift replaces the search randomness.}
For a $\lambda$-biased multiset $S\subseteq\F_2^n$ and a set $C$ of density
$\rho>0$, the averaging operator contracts mean-zero functions by $\lambda$.
Applying it to $\ind_C-\rho$ gives
\begin{equation}
 \Pr_{v\sim U_n}[(v+S)\cap C=\varnothing]
 \le\lambda^2(1-\rho)/\rho.
 \label{m:eq-translate}
\end{equation}
Use explicit $S_j$ of bias $2^{-j/2}$ and size $O(n^2 2^j)$ for $1\le j\le n$
\cite{AGHP92}. Search $v+S_1,\ldots,v+S_n$ with one common shift, followed by
exhaustive fallback. Completeness makes every shift exact. Reaching stage $j\ge2$
requires missing the previous translate, with probability at most
$\min\{1,2^{-(j-1)}/\rho(x)\}$. Charge each reached stage its full size, including
its lazily performed finite-field setup. The expected charge is $O(n^3/\rho(x))$,
hence $O(n^4)$ over inputs and shifts by \cref{m:eq-reciprocal}. Fix a good shift;
two anchors, that shift, and flags cost $3n+O(1)$ bits. Setup is bounded by stage
size times a polynomial, so actual mean time is controlled. A timeout $Q(n)n^c$
gives abstention $n^{-c}$ without changing advice. Stages need not be independent.

\begin{theorem}[Errorless polynomial-time heuristics]\label{m:heuristic}
For every $L\in\Psel$ and fixed positive integer $c$, deterministic polynomial-time
errorless computation with uniform abstention at most $n^{-c}$ uses
$n+(2c+2)\log n+O(1)$ ordinary advice bits. The refined bound is
$n+(c+2)\log n+\log\log n+O_c(1)$; here advice may depend on $c$.
\end{theorem}

\paragraph*{A sample and all its labels from one shift and one cut.}
A one-step witness for positive $x$ is a positive $y$ with $f(x,y)=x$; for negative
$x$ it is a negative $y$ with $f(x,y)=y$. Let $W_x$ be this set and
$w(x)=\mu(W_x)$. Pair counting gives $\mu\{x:w(x)\le t\}\le4t$.
For $\mu_n=G_n(U_r)$, use a bias-$\varepsilon/8$ multiset of
$M=O(r^2/\varepsilon^2)$ seeds and sample $G_n(v+s)$. Inputs with
$w(x)\le\varepsilon/8$ have mass at most $\varepsilon/2$; others are missed with
probability at most $\varepsilon/8$ over $v$, by \cref{m:eq-translate} on seed
space. Some shift has average abstention below $\varepsilon$.

Remove duplicates and construct a directed Hamiltonian path by insertion.
Its labels are a positive prefix followed by a negative suffix, since an adjacent
$0,1$ transition contradicts selectivity. One cut index encodes all labels in
$\lceil\log(M+1)\rceil$ bits. Regenerate the sample, its path, and labels, then
return any sound one-step conclusion, or $\unk$. Shift and cut cost
$r+2\log r+2\log(1/\varepsilon)+O(1)$ bits. Take $r=n$ and
$\varepsilon=n^{-c}$. Construction is polynomial because $M$ is polynomial.
Integrating the tail instead of splitting at a threshold gives
$M=O(r^2(1+\log(1/\varepsilon))/\varepsilon)$ and the refined bound; see
\cref{prop:heuristic-refined}. Only the path's consecutive edges are used, not
transitivity of the selector. Soundness holds even outside the input
 distribution's support.

\section{Black-box lower bounds}
\label{m:lower}
\label{sec:lower}

\subsection{The model}\label{m:model}

A legal length-$n$ instance is a tournament $\tau$ on $X=\bits^n$ and a set
$P\subseteq X$ beating its complement. An interpreter receives $x$ and advice
$a$, queries whether arbitrary $u,v\in X$ satisfy $u\to v$, and has no other
access to $P$. Each input is a fresh run. Oracle-dependent preprocessing and
cached information count as advice; internal oracle-independent computation
may be arbitrarily large but finite. A $q$-query bound holds on every input,
permitted advice, tournament, and random tape. Randomized procedures have a
finite coin bound at each length, which may be padded. Self-comparisons have
a fixed convention. For a bit-valued $M$, call $a$ bounded-error correct if
\begin{equation}
 \Pr_R[M^\tau(x,a;R)=\ind[x\in P]]\ge2/3\quad\text{for every }x\in X.
 \label{m:eq-good}
\end{equation}
Only the designated advice must be correct; other strings need not have an
acceptance gap. For outputs $0,1,\unk$, sound advice never gives a wrong bit
on any input or tape. Write
\begin{equation}
 \operatorname{cov}_{\tau,P}(H,a)=\Pr_{x\sim U_n,R}[H^\tau(x,a;R)\ne\unk].
 \label{m:eq-coverage}
\end{equation}

An edge query and a query to the mixed-vector comparator induced by that
selector simulate each other with one call. Thus these lower bounds already
apply to comparator-based procedures on the P-selective subclass; the
selector-specific average-case upper bounds are not extended to all $\Mc$.

\subsection{Transcript freezing for errorless heuristics}\label{m:freeze}

\begin{theorem}[Transcript freezing]\label{m:freezing}
For every $q$-query heuristic, integers $q,\ell\ge0$, and $n\ge1$, there is a
legal instance with a transitive tournament such that every sound advice of
length at most $\ell$ has
\begin{equation}
 \operatorname{cov}_{\tau,P}(H,a)
 \le\min\left\{1,\frac{(2q+1)(2^{\ell+1}-1)}{2^n}\right\}.
 \label{m:eq-freezing}
\end{equation}
Indeed, one set of at least $2^n-(2q+1)(2^{\ell+1}-1)$ inputs forces every
sound advice string to abstain on every tape.
\end{theorem}

\begin{proof}
Represent a transitive tournament by a total order, with earlier vertices
beating later ones. Maintain three consecutive blocks $F_+\ U\ F_-$.
Vertices in $F_+$ are permanently positive, those in $F_-$ permanently
negative, and those in $U$ uncommitted. The relative order of all permanently
fixed vertices is frozen. Initially $F_+=F_-=\varnothing$ and $U=X$ is in
lexicographic order. Keep all $A=2^{\ell+1}-1$ eligible advice strings live.

Examine the current full tournament. If every live advice returns $\unk$ on
every input in $U$, on every tape, stop. Otherwise choose live advice $a$,
$x\in U$, and a tape $R$ on which the run returns a bit $b$. Let $F$ consist
of $x$ and all queried endpoints belonging to $U$; then $|F|\le2q+1$.
If $b=0$, remove $F$ from $U$, declare it positive, and insert it immediately
after the old $F_+$. If $b=1$, declare it negative and insert it immediately
before the old $F_-$. Preserve the internal order of $F$, the order of the
remaining $U$, and the order of previously fixed vertices. Remove $a$ from
the live set.

Every endpoint of the selected run is now fixed permanently. Comparisons
within $F$ are unchanged, and comparisons with old fixed vertices are unchanged
because each new vertex remains below all old positives and above all old
negatives. Thus every answer in this run is preserved, including at all later
stages. Once $R$ is fixed the procedure is deterministic, so its final
transcript and output are still $b$, while $x$ has permanent label $1-b$.
This advice is forever unsound. A selected finite random tape has positive
probability, which suffices to violate errorlessness.

Each nonterminal stage eliminates one advice and fixes at most $2q+1$
vertices. There are at most $A$ stages. At termination declare the remaining
$U$ negative without changing the order. The final positive set is $F_+$,
a prefix, so the instance is legal. Eliminated advice remains unsound; every
live advice abstains on all of $U$ on every tape, and
$|U|\ge2^n-(2q+1)A$. This proves the claims. We reexamine live advice after
each update: only selected erroneous transcripts, not earlier abstentions,
are assumed to remain unchanged.
\end{proof}

\begin{corollary}[Tight leading linear term]\label{m:near-n}
If every legal instance admits sound advice of length at most $\ell$ giving
coverage $1-\varepsilon$, then
\begin{equation}
 \ell\ge\log\left(1+\frac{(1-\varepsilon)2^n}{2q+1}\right)-1.
 \label{m:eq-near-n}
\end{equation}
For polynomial $q$ and $\varepsilon\le1/2$, this is $n-O_q(\log n)$.
\end{corollary}

The corollary follows by rearranging \cref{m:eq-freezing}. It matches
\cref{m:heuristic} up to $O(\log n)$. Truncating an exact decider at twice its
polynomial mean bound gives a sound polynomial-query heuristic of coverage
$1/2$, so the same linear lower order applies to exact mean time. Freezing
does not prove a bounded-error average-case lower bound: its preserved error
may occur on one input and one very unlikely tape.

\subsection{The hidden-core distribution and its deterministic bound}
\label{m:core}

Write $x=(\beta,w)$, with a bucket $\beta\in\bits^m$ and a tail
$w\in\bits^{n-m}$. Independently choose one uniform tail $w_\beta$ and label
$\lambda_\beta\in\bits$ per bucket, and one fair orientation coin
$c_{\beta\gamma}$ per unordered pair of buckets. The core is
$\Core=\{t_\beta=(\beta,w_\beta)\}$ and
$P=\{t_\beta:\lambda_\beta=1\}$. Outside the core, the lexicographically
smaller vertex wins. Every core vertex beats every noncore vertex. Within
the core, opposite labels make the positive vertex win; equal labels use
the pair's coin. This always defines a legal instance.

\begin{proposition}[Uniform deterministic distributional estimate]
\label{m:delta}
For every deterministic $q$-query procedure $D$, if $4q\le2^{n-m}$ and
$k\ge1$, its probability of deciding the entire instance is at most
\begin{equation}
 \Delta(n,m,k,q)=
 2^{m-k(n-m-\log(4q))}
 +(1-2^{-k})^{\lfloor2^m/k\rfloor}.
 \label{m:eq-delta}
\end{equation}
The distribution and bound are independent of the procedure, its description
length, and any fixed advice or tape incorporated in its code.
\end{proposition}

\paragraph*{Proof sketch.}
On input $t_\beta$, define an ambiguous run by answering an edge to another
core vertex $t_\gamma$ as a loss when $\lambda_\gamma=1$ and a win when
$\lambda_\gamma=0$. All other answers are unchanged. Let $E_\beta$ be the
other buckets whose core point occurs as a query endpoint, $D_\beta\subseteq
E_\beta$ those actually compared with $t_\beta$, and $o_\beta$ the output.
The run reads neither $\lambda_\beta$ nor any incident coin. Conditional on
all other variables, it agrees with the true run and errs with probability
at least $\tfrac12 2^{-|D_\beta|}$: make $\lambda_\beta\ne o_\beta$ and
require the equal-label incident coins to point in the ambiguous direction.

There are two estimates. Exposing actual discoveries, even if the procedure
cannot recognize them, leaves an undiscovered tail uniform over at least
$2^{n-m}-2q\ge2^{n-m-1}$ possibilities. Each of at most $2q$ endpoints thus
has discovery probability at most $2^{1-(n-m)}$. A union bound gives
\[
 \Pr[\max_\beta|E_\beta|\ge k]
 \le2^{m-k(n-m-\log(4q))}.
\]
For fixed tails, expose ambiguous runs in rounds, always starting at an
unused bucket. Keep the invariant that all exposed labels and both endpoints
of all exposed coins are used. A run discovering fewer than $k$ other buckets
uses at most $k$ buckets after its input is included. The input's label and
incident coins remain fresh until that run finishes, and force an error with
conditional probability at least $2^{-k}$. At least $\lfloor2^m/k\rfloor$
rounds are possible unless many discoveries or an error has already excluded
the target event. Hence
\[
 \Pr[D\text{ correct on }\Core\text{ and }\max_\beta|E_\beta|<k]
 \le(1-2^{-k})^{\lfloor2^m/k\rfloor}.
\]
Add the estimates. The full exposure argument, including its conditional
independence invariant, is in \cref{lem:discovery,lem:fresh}.

For a fixed advice budget, a union bound over fewer than $2^{\ell+1}$ strings
turns \cref{m:eq-delta} into a finite lower bound. Balancing
$m\approx(n-\log q)/2$ and $k=m-O(\log n)$ yields coefficient $1/4$.
Importantly, \cref{m:delta} holds before this union bound and uniformly over
all fixed procedures. This is what permits the following randomized transfer.

\subsection{Randomized lifting with no dependence on coin length}
\label{m:lifting}

\begin{lemma}[Yao's principle plus amplification \cite{Yao77}]
\label{m:lift}
Put $t_n=12(n+2)+1$. Suppose a distribution $\mathcal D$ on legal length-$n$
instances satisfies
\begin{equation}
 \Pr_{I\sim\mathcal D}[D\text{ decides }I]\le\delta
 \label{m:eq-lift-hyp}
\end{equation}
for every deterministic procedure with at most $t_nq$ queries, with no bound
on its description length or internal computation. For every randomized
$q$-query procedure $M$ with any finite coin bound, and each fixed ordinary
advice $a$,
\begin{equation}
 \Pr_{I\sim\mathcal D}[a\text{ is bounded-error correct for }M]
 \le\frac43\delta.
 \label{m:eq-lift-conc}
\end{equation}
\end{lemma}

\begin{proof}
Run $M$ independently $t_n$ times on the same $x,a$ and take the majority.
Let $R$ concatenate the tapes and let $A(x,a;R)$ be the amplified output.
Fix an instance on which $a$ satisfies \cref{m:eq-good}. If $Z$ counts
incorrect repetitions for a fixed $x$, independence gives
$\E 2^Z\le(4/3)^{t_n}$. Thus
\begin{equation}
 \Pr_R[A(x,a;R)\ne\ind[x\in P]]
 \le2^{-t_n/2}(4/3)^{t_n}
 =(8/9)^{t_n/2}<2^{-n-2},
 \label{m:eq-amplify}
\end{equation}
since $(8/9)^6<1/2$ and $t_n/2>6(n+2)$. A union bound over the $2^n$
inputs shows that with probability at least $3/4$ over one shared $R$,
$A(x,a;R)$ is correct for every input. Independence between different inputs
is unnecessary.

For every fixed $R$, $D_{a,R}(x)=A(x,a;R)$ is deterministic and uses at most
$t_nq$ queries. Its description may contain $a$ and all of $R$, but neither
depends on the random instance, so \cref{m:eq-lift-hyp} applies uniformly.
Let $G_a$ be the event that $a$ is bounded-error correct, and let $F(I,R)$
indicate simultaneous correctness of $D_{a,R}$ on $I$. Interchanging the two
independent averages gives
\begin{equation}
 \frac34\Pr_I[G_a]
 \le\E_I\E_R F(I,R)
 =\E_R\Pr_I[D_{a,R}\text{ decides }I]
 \le\delta.
 \label{m:eq-lift-average}
\end{equation}
This proves the claim. We neither select an instance-dependent tape and
treat it as free advice, nor take a union bound over tape space. This is the
distributional direction of Yao's principle after amplification; the uniform
deterministic estimate makes its only cost an $O(n)$ factor in queries.
\end{proof}

\begin{theorem}[Quadratic advice, independently of the coins]
\label{m:quadratic}
There is an absolute constant $C_{\mathrm R}$ such that, whenever
$1\le q\le2^{n-4}$ and
\begin{equation}
 B_{\mathrm R}(n,q)=
 \left\lfloor\frac{(n-\log q)^2}{4}-C_{\mathrm R}n\log(n+2)\right\rfloor
 \ge0,
 \label{m:eq-BR}
\end{equation}
every randomized $q$-query decider with any finite coin bound has a legal
instance on which no ordinary advice of length at most $B_{\mathrm R}(n,q)$
gives error at most $1/3$ on every input. For polynomial $q$ the lower bound
is $n^2/4-O_q(n\log n)$; for $q\le2^{\varepsilon n}$ with fixed
$0<\varepsilon<1$, it is $(1-\varepsilon)^2n^2/4-O(n\log n)$.
\end{theorem}

\paragraph*{Parameter accounting.}
Use \cref{m:delta,m:lift} with $Q=t_nq$. If
\[
 \ell+4\le k(n-m-\log(4Q))-m,\qquad
 \ell+4\le\lfloor2^m/k\rfloor2^{-k}\log e,
\]
and $4Q\le2^{n-m}$, then $\Delta\le2^{-\ell-3}$. The probability that any
eligible advice is bounded-error correct is less than
$2^{\ell+1}(4/3)2^{-\ell-3}=1/3$. Take
$m=\lfloor(n-\log Q)/2\rfloor$ and $k=m-\lceil3\log n\rceil-2$.
The finite inequalities yield
$\ell\ge(n-\log Q)^2/4-O(n\log n)$. Since $\log Q=\log q+O(\log n)$,
the leading coefficient is unchanged. \Cref{thm:randomized-finite,cor:randomized-quadratic}
verify the excluded and nonvacuous parameter ranges. Amplification also
preserves this coefficient for any specified inverse-polynomial pointwise
advantage over $1/2$.

\subsection{The leading constant: queries versus polynomial time}
\label{m:constant}

Ko's deterministic upper bound stores at most $n$ positive anchors and a
presence mask, for $n^2+n$ bits and $n$ comparisons \cite{Ko83}. Repeatedly
choosing a minimum-outdegree vertex in the residual positive tournament and
removing it with all vertices that beat it halves the residual size.
Decoding this advice is polynomial-time.

\begin{theorem}[Half-quadratic advice in the query model]\label{m:half}
With unrestricted finite oracle-independent computation, a single
deterministic decoder decides every legal instance using $n+2$ comparisons
and $n^2/2+O(n\log n)$ ordinary advice bits.
\end{theorem}

The proof counts many near-halving greedy choices. A uniform tuple of at most
$n+2$ vertices is a valid positive anchor tuple with probability at least
$2^{-n^2/2-O(n\log n)}$. A union bound over all tournaments produces a
universal family of $2^{n^2/2+O(n\log n)}$ tuples. Advice indexes a valid tuple
in the lexicographically first such family. Exhaustive reconstruction is
finite and does not query the actual oracle, but is \emph{not} shown to be
polynomial-time. The full construction is in \cref{sec:constant}. This improves
only the query-model upper coefficient, not Ko's polynomial-time coefficient.

The single-core family also has a typical-instance $n^2/4+O_d(n)$ upper
bound at $m=n/2$: a mask for the first $4m$ buckets identifies the first
$m+d$ positive tails to advise. It succeeds on a fraction at least
$1-2^{-d}-e^{-\Omega(m)}$ for fixed $d$ (\cref{prop:typical}). This does not
bound rare instances. Naively nesting cores does not add advice costs: an
inner positive anchor certifies outer positives, and an inner negative
anchor certifies outer negatives. A stronger multiscale lower bound must
prevent this reuse.

\section{One oracle realizing both tight orders}
\label{m:oracles}

For $n\ge2$, define
\begin{equation}
 g(n)=\max\{0,\lfloor n^2/4-n(\log n)^2\rfloor\},\qquad
 h(n)=\max\{0,\lfloor n-(\log n)^2\rfloor\};
 \label{m:eq-budgets}
\end{equation}
set both to zero at $n=0,1$. We use the more permissive designated-advice
convention: $L\in\BPP^W/g$ means that a polynomially clocked probabilistic
oracle machine and advice $|a_n|\le g(n)$ satisfy
$\Pr_R[M^W(x,a_n;R)=\chi_L(x)]\ge2/3$ for every $x$. The clock applies to
all permitted advice and tapes, but no acceptance gap is required on
nondesignated advice. A lower bound here also applies to the stricter
all-advice bounded-gap convention.

\begin{theorem}[A joint oracle separation]\label{m:joint}
There exist $W$ and $L\in\PselW$ such that:
\begin{enumerate}[(i)]
\item $L\notin\BPP^W/g(n)$, allowing any polynomial number of coins;
\item every polynomial-time randomized $W$-oracle errorless heuristic has
a length at which all sound advice strings of length at most $h(n)$ have
uniform coverage less than $1/2$;
\item the same language has all the upper bounds of
\cref{m:det,m:random,m:heuristic} relative to $W$, including the randomized
bound for every exactly polynomial-time $W$-samplable distribution.
\end{enumerate}
Consequently exact polynomial uniform mean time has advice order $\Theta(n)$
on this language, while worst-case bounded-error polynomial time has advice
order $\Theta(n^2)$.
\end{theorem}

\paragraph*{Construction sketch.}
Interleave worst-case and heuristic requirements, repeating each program at
arbitrarily large clock exponents. At stage $i$, use clock $n^i$ and an active
length $n_i$, with $n_{i+1}=2^{n_i}$. Choose $n_1$ large enough that
\begin{equation}
 n_i^i\le2^{n_i-4},\quad g(n_i)\le B_{\mathrm R}(n_i,n_i^i),\quad
 \frac{(2n_i^i+1)(2^{h(n_i)+1}-1)}{2^{n_i}}<\frac12.
 \label{m:eq-stage}
\end{equation}
The first condition enforces \cref{m:quadratic}'s query range. Put a legal
instance at each active length and make $L$ empty elsewhere. Encode its edges
at length $2n_i+2$ and membership at length $2^{n_i}+2$. A clocked machine at
$n_i$ can see neither its membership records nor any later records. All
accessible answers except the current edges are already fixed, reducing it
to a $q=n_i^i$ query procedure. Use \cref{m:quadratic} at worst-case stages
and \cref{m:freezing} at heuristic stages. Freeze all current records.

The global selector uses edge records for equal active lengths. If just one
length is active it chooses that input. For two active lengths $r<s$, the
shorter input's membership record has length $2^r+2\le s+2$, so the selector
can read it and choose a positive shorter input, or otherwise the longer
input, in polynomial combined length. This completes the global selectivity
argument; same-length hard tournaments alone would not suffice. All mean-time
and heuristic upper bounds relativize to this $W$. Truncation at twice a
putative polynomial mean gives coverage $1/2$, proving the exact mean-time
lower bound. The full clocks, encodings, and diagonalization are in
\cref{sec:oracle}. Hardness is forced at active lengths, not at every length.

\subsection{Opposing relativizations and the two advice models}

For every oracle $A$, \cref{m:mc-rel} gives $\Mc^A$ linear nondeterministic
advice; the two-step positive-anchor verifier improves this on its subclass to
$\mathrm P^A\text{-}\mathrm{sel}\subseteq\NP^A/(n+1)$, with $n$ witness bits
\cite{HT96,HNP98}. With a PSPACE-complete oracle, $\Pclass^A=\NP^A=\PSPACE$,
so deterministic $n+1$ advice bits suffice for selectors, and $3n+5$ for
binary comparators. This opposes \cref{m:joint}:
neither deterministic nor bounded-error containment of $\Psel$ or $\Mc$
with $O(n)$ or $O(n^{2-\varepsilon})$ advice, $0<\varepsilon<1$, can be settled
by a proof valid relative to every oracle. By \cref{m:mc-U}, an unrelativized
counterexample would refute $\mathsf U$, and hence imply $\Pclass\ne\NP$;
a bounded-error counterexample would also imply $\NP\not\subseteq\BPP$.

In the coin-dependent model, $a_{n,r}$ may depend on the coins $r$, but never
on the input. A minimax distribution on valid anchor pairs covers every input
with constant probability; a list of $O(n)$ constant-size tuples of such
pairs covers every input on at least two thirds of its indices. The
interpreter receives only the chosen tuple, not the whole list. This gives
an errorless $O(n)$-advice interpreter using $\log n+O(1)$ coins
\cite{BD12}; a standalone argument appears in \cref{prop:dependent-advice}.

\begin{corollary}[Ordinary versus coin-dependent advice]\label{m:models}
For the same $W,L$, $L\notin\BPP^W/g(n)$ but $L\in\BPP^W//O(n)$.
The latter has an errorless interpreter with $\log n+O(1)$ coins and conclusive
probability at least $2/3$ on every input. Polynomially many coins do not
remove the quadratic ordinary-advice requirement.
\end{corollary}

\section{What non-black-box access to the comparison program buys}
\label{m:conditional-section}

\paragraph*{Hypothesis $\mathsf U$.}
There is a deterministic algorithm, polynomial-time on every Boolean circuit,
that rejects circuits with no satisfying assignment and accepts circuits with
exactly one; either answer is permitted on circuits with more solutions.
This is promise Unique-Circuit-SAT, not the total exact-one language or
$\Pclass=\mathrm{UP}$. The hypothesis follows from $\Pclass=\NP$ and implies
$\NP=\mathrm{RP}$ by isolation \cite{VV86}; no strict separation of these
hypotheses is claimed.

\begin{theorem}[Advice-preserving isolation]\label{m:generic}
Assume $\mathsf U$. Let $a(n),w(n)$ be polynomially bounded and polynomial-time
computable. If a polynomial-time predicate $R$ and ordinary advice $a_n$ of
length $a(n)$ satisfy
$x\in L\Longleftrightarrow\exists z\in\bits^{w(n)}R(x,a_n,z)=1$,
then $L\in\Pclass/O(a(n)+w(n)+n)$.
\end{theorem}

\paragraph*{Proof idea.}
Compile $R(x,a_n,z)$ into a circuit $C$ with exactly $m=w(n)$ free inputs;
internal gates introduce no witness variables. For $m\ge1$, a random Toeplitz
matrix $T\in\F_2^{(m+1)\times m}$ and shift $b\in\F_2^{m+1}$ define a pairwise
independent hash $h_s(z)=Tz+b$ with $3m+1$ seed bits. For every
$k=1,\ldots,m+1$, apply the assumed solver to
$C(z)\wedge[h_s(z)_{1\ldots k}=0^k]$ and OR its answers. Unsatisfiable $C$ is
always rejected. If $C$ has $M>0$ solutions, take $k=\lceil\log M\rceil+1$.
The survivor count $Z$ has $\mu=\E Z\in(1/4,1/2]$ and
$\E[Z(Z-1)]\le\mu^2$, hence
\[
 \Pr[Z=1]\ge\E Z-\E[Z(Z-1)]\ge\mu-\mu^2\ge3/16.
\]
Testing all prefixes avoids a factor-$m$ loss. Off-promise acceptances are
harmless because the original circuit is then satisfiable.

For each positive $n$-bit input, good seeds occupy at least $3/16$ of seed space.
A stationary length-$O(n)$ walk on a fixed constant-degree explicit expander
misses each such set with probability below $2^{-n-2}$ \cite{RVW02}. A union bound
provides one walk hitting all at most $2^n$ good sets. Its initial vertex and edge
labels cost $3m+1+O(n)$ bits, rather than quadratic advice for independent seeds.
Advise this walk with $a_n$; reconstruct it and OR the tests. All negative inputs
reject on every seed and all positives accept on some visited seed. The case
$m=0$ is direct. Full hashing, spectral, and running-time details are in
\cref{sec:conditional}.

\Cref{m:mc-U} is the direct substitution of the comparator verifier's
$a(n)=3n+5$ and $w(n)=5n+12$. In particular, it also yields the earlier
selector consequence $\Psel\subseteq\Pclass/O(n)$ under
$\mathsf U$. This is not a black-box simulation: it compiles the actual selector
or comparator program into ordinary circuits. A comparison oracle supplies no
such circuit, and $\mathsf U$ does not solve promise-unique circuits containing
arbitrary oracle gates. Nor does $\Pclass=\mathrm{UP}$ justify these calls, since
restricted circuits may be off-promise. A generic $\NP=\mathrm{RP}$ simulation
does not preserve seed length as a function of witness length.

\section{Open problems}\label{m:open}\label{sec:open}

\paragraph*{Bounded-error heuristics.}
Freezing eliminates advice by preserving one erroneous run, even on a negligible
set of inputs and tapes. What advice lower bounds hold when both error and
abstention are measured under the uniform distribution?

\paragraph*{Constants and decoding time.}
The query-model gap is $1/4$ versus $1/2$; polynomial-time decoding still has
Ko's upper coefficient $1$. Can the lower bound increase, or can efficient
indexing beat coefficient $1$? Inner anchors can serve multiple nested layers.
For binary nondeterministic advice, can coefficient $3$ or the certificate and
call constants be reduced? Higher-arity comparators require more than the
one-premise implication behind the short-forcing proof.

\paragraph*{Weakening $\mathsf U$.}
The transfer uses the promise solver only on hash restrictions of circuits
compiled from comparison verifiers. Could a weaker hypothesis suffice with
linear seed length? Neither $\Psel\subseteq\Pclass/O(n)$ nor the analogous
unconditional containment for $\Mc$ is established here.

\paragraph*{Research methods and AI assistance}
OpenAI ChatGPT was used for substantive mathematical assistance: developing and
revising the weighted-certificate algorithms, transcript-freezing argument,
randomized lifting, query-only covering-family bound, joint oracle construction,
and advice-preserving isolation transfer; checking parameter calculations; and
locating literature. The human-led discussion supplied the hidden-core
construction and small-bias improvements. For the supplied binary-comparator
draft, ChatGPT condensed and retypeset the proof, checked its advice, witness,
and query accounting, and integrated its consequences. It also prepared and
reorganized LaTeX. \Cref{app:ai} records specific uses and provenance limitations.
No AI system is an author; the human author(s) remain responsible for the claims,
proofs, and references.
\label{m:last-page}

\par\bigskip
\begingroup
\small
\phantomsection
\addcontentsline{toc}{section}{References}
\bibliographystyle{plainurl}
\bibliography{linear_certificates_quadratic_barriers_arxiv}
\endgroup
\clearpage
\appendix
\crefalias{section}{appendix}
\section{Tournament lemmas and certificate completeness}
\label{sec:preliminaries}\label{app:prelim}

All logarithms are base two except $\ln$. Asymptotic bounds concern
$n\ge2$; finitely many shorter slices can be handled separately. Ordinary
advice is fixed simultaneously for all inputs of a length and independently
of random coins; it need not be computable. Time bounds are measured with
the designated advice unless a clock is expressly part of the model.
Polynomial mean time means an expectation at each length, not a pointwise
polynomial bound.

\subsection{Conventions}

A language $L\subseteq\bits^*$ is P-selective if it has a polynomial-time
computable function $f$ satisfying
\begin{equation}
 f(x,y)\in\{x,y\},\qquad
 \{x,y\}\cap L\ne\varnothing\Longrightarrow f(x,y)\in L.
 \label{eq:selector-definition}
\end{equation}
We may make $f$ commutative by first ordering its two arguments
lexicographically. This transformation preserves
\eqref{eq:selector-definition} and polynomial running time.

For distinct strings of the same length, write
\begin{equation}
 x\to y\quad\Longleftrightarrow\quad f(x,y)=x.
 \label{eq:orientation}
\end{equation}
There are no loops. Thus $\to$ is a tournament orientation. Fixing $n$,
put
\[
 X=\bits^n,\qquad P=L\cap X,\qquad N=X\setminus L.
\]
Every vertex in $P$ beats every vertex in $N$.

For a nonnegative weight function $\nu$ on a finite set, write
$\nu(A)=\sum_{x\in A}\nu(x)$. A weight function used on an induced
subtournament need not sum to one. In particular, when a distribution on
$X$ is restricted to $P$ or $N$, we do not renormalize it.

\subsection{A weighted tournament lemma}

\begin{lemma}[Lower tail of two-step witness weight]
\label{lem:weighted}
Let $V$ be a finite tournament with nonnegative vertex weights $\nu$.
Choose $p\in V$ minimizing the weighted outdegree
\[
 d^+(p)=\sum_{z:p\to z}\nu(z).
\]
For $x$ with $p\to x$, define
\[
 I_x=\{z:x\to z\to p\},\qquad
 J_x=\{z:p\to z\to x\}.
\]
Then
\begin{equation}
 \nu(I_x)\ge\nu(J_x)+\nu(x).
 \label{eq:pointwise-weight}
\end{equation}
For every $t\ge0$,
\begin{equation}
 \nu\{x:p\to x,\ \nu(I_x)\le t\}\le2t.
 \label{eq:weighted-tail}
\end{equation}
\end{lemma}

\begin{proof}
Partition $V\setminus\{p,x\}$ according to the two orientations involving
$p$ and $x$. The common out-neighbors cancel in the degree difference;
the edge $p\to x$ contributes $-\nu(x)$. Consequently,
\begin{equation}
 d^+(x)-d^+(p)=\nu(I_x)-\nu(J_x)-\nu(x).
 \label{eq:degree-difference}
\end{equation}
The left-hand side is nonnegative by the choice of $p$, which proves
\eqref{eq:pointwise-weight}.

Let $A_t=\{x:p\to x,\ \nu(I_x)\le t\}$. If $x,y\in A_t$ and
$y\to x$, then $y\in J_x$, since $p\to y$. Therefore
\begin{align*}
 t\nu(A_t)
 &\ge\sum_{x\in A_t}\nu(x)\nu(I_x)\\
 &\ge\sum_{x\in A_t}\nu(x)\bigl(\nu(J_x)+\nu(x)\bigr)\\
 &\ge\sum_{\substack{x,y\in A_t\\y\to x}}\nu(x)\nu(y)
       +\sum_{x\in A_t}\nu(x)^2\\
 &=\frac12\left(\nu(A_t)^2+\sum_{x\in A_t}\nu(x)^2\right)\\
 &\ge\frac12\nu(A_t)^2.
\end{align*}
The equality counts each unordered pair once. If $\nu(A_t)=0$, the
claim is immediate; otherwise division by $\nu(A_t)$ proves
\eqref{eq:weighted-tail}.
\end{proof}

\begin{corollary}[Unweighted form]
\label{cor:unweighted}
With unit vertex weights and $k\ge1$ an integer,
\begin{equation}
 \bigl|\{x:p\to x,\ |I_x|\le k\}\bigr|\le2k-1.
 \label{eq:unweighted-tail}
\end{equation}
Every vertex beaten by $p$ has a two-step path to $p$.
\end{corollary}

\begin{proof}
The second claim follows from \eqref{eq:pointwise-weight}.
For the first, set $A=\{x:p\to x,\ |I_x|\le k\}$ and $s=|A|$.
The preceding proof, retaining the diagonal term, gives
$ks\ge(s^2+s)/2$. If $s>0$, this implies $s\le2k-1$.
\end{proof}

\subsection{Two advised vertices}

Let $\nu$ now be a full-support probability distribution on $X$. Assume
first that $P$ and $N$ are both nonempty. Choose
\begin{align}
 p&\in\operatorname*{arg\,min}_{u\in P}
       \sum_{z\in P:u\to z}\nu(z),\label{eq:positive-anchor}\\
 q&\in\operatorname*{arg\,min}_{u\in N}
       \sum_{z\in N:z\to u}\nu(z).
 \label{eq:negative-anchor}
\end{align}
These choices are nonuniform. Computing the minima is not part of the
interpreter's task. The advice stores $p,q$ and a constant-size flag
indicating whether either class is empty. Unused fields may be padded,
so the advice length is $2n+O(1)$.

On input $x$, the \emph{direct tests} are
\begin{align}
 x=p\ \text{or}\ x\to p&\quad\Longrightarrow\quad\text{output }1,
 \label{eq:positive-direct}\\
 x=q\ \text{or}\ q\to x&\quad\Longrightarrow\quad\text{output }0.
 \label{eq:negative-direct}
\end{align}
Let $\mathcal H$ be the set of inputs not decided by these tests. For
$x\in\mathcal H$, define
\begin{align}
 C_+(x)&=\{z\in X:x\to z\to p\},\label{eq:Cplus}\\
 C_-(x)&=\{z\in X:q\to z\to x\},\label{eq:Cminus}\\
 C(x)&=C_+(x)\cup C_-(x),\qquad
 \rho(x)=\nu(C(x)).\label{eq:rho}
\end{align}
An intermediate vertex in $C_+(x)$ certifies that $x$ is positive; one
in $C_-(x)$ certifies that $x$ is negative. Testing a proposed certificate
uses a constant number of evaluations of $f$.

\begin{proposition}[Soundness, completeness, and density profile]
\label{prop:anchors}
With the designated advice, all the direct tests and two-step
certificates are sound. For every $x\in\mathcal H$,
\begin{equation}
 \rho(x)\ge\nu(x)>0.
 \label{eq:rho-atom}
\end{equation}
Moreover, for every $t\ge0$,
\begin{equation}
 \nu\{x\in\mathcal H:\rho(x)\le t\}\le4t.
 \label{eq:rho-tail}
\end{equation}
\end{proposition}

\begin{proof}
Because $p$ is positive, $x\to p$ forces $x$ to be positive. Similarly,
since $q$ is negative, $q\to x$ forces $x$ to be negative. This proves
soundness of the direct tests.

If $x\to z\to p$, then $z$ is positive, and hence so is $x$.
If $q\to z\to x$, then $z$ is negative, and hence so is $x$.
Thus the two certificate types cannot conflict on any input.

For $x\in\mathcal H\cap P$, the failure of the positive direct test
implies $p\to x$. A vertex $z$ with $z\to p$ must belong to $P$.
Therefore $C_+(x)$ is exactly the set $I_x$ in the tournament induced
by $P$, and $C_-(x)=\varnothing$. Apply \cref{lem:weighted} to that
subtournament, using the original weights $\nu$.

For $x\in\mathcal H\cap N$, the failure of the negative direct test
implies $x\to q$. Reverse all edges in the tournament induced by $N$.
The choice of $q$ in \eqref{eq:negative-anchor} is a minimum-outdegree
choice in the reversed tournament, and its two-step paths correspond
exactly to $q\to z\to x$ in the original orientation. The lemma
therefore applies to this class as well.

The pointwise part of the lemma gives \eqref{eq:rho-atom} in both
classes. Each class contributes at most $2t$ to the lower tail; adding
the two bounds gives \eqref{eq:rho-tail}. If $P$ or $N$ is empty, the
flag decides the whole slice, and the conclusions hold with
$\mathcal H=\varnothing$.
\end{proof}

\begin{remark}[No labels for intermediate vertices]
The advice stores only the two anchors. It does not store the membership
labels of sampled intermediate vertices: the edges $z\to p$ and
$q\to z$ supply precisely the label information needed for soundness.
\end{remark}

\section{Full proofs of the average-case upper bounds}
\label{sec:upper}\label{app:upper}

The following statements concern every unrelativized P-selective language.
Their interpreters use the selector only through its evaluations, so the
proofs also apply relative to an arbitrary oracle. The running-time
qualifications distinguish exact mean-time computation from a
worst-case polynomial-time errorless heuristic.

\subsection{Main upper bounds}

\begin{theorem}[Exact deterministic computation with polynomial mean time]
\label{thm:main-deterministic}
For every $L\in\Psel$, there are a deterministic machine $D$, a polynomial
$Q$, and ordinary advice strings $a_n$ such that
\begin{align}
 |a_n|&\le3n+O(1),\label{eq:main-advice}\\
 D(x,a_n)&=\chi_L(x)\quad\text{for every }x\in\bits^n,\label{eq:main-exact}\\
 \E_{x\sim U_n}T_D(x,a_n)&\le Q(n).\label{eq:main-mean}
\end{align}
The worst-case time is $2^n\poly(n)$. The mean number of certificate
tests, and the mean stage cost before polynomial per-operation overhead,
are $O(n^4)$.
\end{theorem}

\begin{theorem}[Exact randomized computation under samplable distributions]
\label{thm:main-randomized}
Let $L\in\Psel$ and let $\mu_n$ be exactly samplable in worst-case
polynomial time. There is a randomized interpreter with $2n+O(1)$ ordinary
advice bits that never gives an incorrect answer, terminates almost
surely on every input, and satisfies
\begin{equation}
 \E_{x\sim\mu_n}\E_\omega T_R(x,b_n;\omega)\le\poly(n).
 \label{eq:main-random-mean}
\end{equation}
Its mean number of sampled certificates is $O(n)$. Truncation after $M$
samples gives average abstention at most $8/(M+1)$, or $4/(M+1)$ when
using the direct uniform-distribution construction.
\end{theorem}

\begin{theorem}[Polynomial-time errorless heuristics]
\label{thm:main-heuristic}
For every $L\in\Psel$ and every fixed positive integer $c$, a deterministic
polynomial-time interpreter using
\begin{equation}
 n+(2c+2)\log n+O(1)
 \label{eq:main-heuristic-advice}
\end{equation}
ordinary advice bits gives only answers in $\{\chi_L(x),\unk\}$ and has
uniform abstention probability at most $n^{-c}$. A refinement uses only
\begin{equation}
 n+(c+2)\log n+\log\log n+O_c(1)
 \label{eq:refined-heuristic-advice}
\end{equation}
bits. The advice may depend on $c$. The $3n+O(1)$ advice sequence in
\cref{thm:main-deterministic}, by contrast, supports every fixed
inverse-polynomial timeout guarantee without changing the advice.
\end{theorem}

\subsection{Exact randomized search}
\label{sec:randomized}

The following reciprocal-density estimate is the link between the
combinatorial bound and average running time.

\begin{lemma}[Mean reciprocal certificate density]
\label{lem:reciprocal}
Under the hypotheses of \cref{prop:anchors}, suppose $\nu(x)\ge\beta>0$
for every $x\in X$. Set
\[
 B(x)=\begin{cases}1/\rho(x),&x\in\mathcal H,\\0,&x\notin\mathcal H.
 \end{cases}
\]
Then
\begin{equation}
 \E_{x\sim\nu}[B(x)]\le1+4\ln(1/\beta).
 \label{eq:reciprocal}
\end{equation}
\end{lemma}

\begin{proof}
By \eqref{eq:rho-atom}, $0\le B(x)\le1/\beta$.
For $u\ge1$, \eqref{eq:rho-tail} implies
\[
 \Pr_{x\sim\nu}[B(x)>u]
 \le\nu\{x\in\mathcal H:\rho(x)\le1/u\}\le4/u.
\]
Integrating the tail gives
\[
 \E B=\int_0^{1/\beta}\Pr[B>u]\,du
 \le1+4\int_1^{1/\beta}\frac{du}{u}
 =1+4\ln(1/\beta).
\]
\end{proof}

\subsubsection*{The certificate-search procedure}

After performing the direct tests, repeatedly sample independent
$z\sim\nu$. Output $1$ if $z\in C_+(x)$, output $0$ if
$z\in C_-(x)$, and otherwise continue. Soundness follows from
\cref{prop:anchors}. For $x\in\mathcal H$, the number $\tau(x)$ of
samples is geometric with success probability $\rho(x)$, so
\begin{equation}
 \E_{\omega}\tau(x)=1/\rho(x).
 \label{eq:geometric}
\end{equation}
Put $\tau(x)=0$ on directly decided inputs. The procedure terminates
almost surely on every input by \eqref{eq:rho-atom}.

\begin{lemma}[Average truncation error]
\label{lem:truncation}
For every integer $m\ge1$, the probability, averaged over $x\sim\nu$
and the independent samples, that the procedure has not answered after
$m$ samples is at most
\begin{equation}
 \frac4{m+1}.
 \label{eq:truncation}
\end{equation}
\end{lemma}

\begin{proof}
For $x\in\mathcal H$, the probability of missing all certificates is
$(1-\rho(x))^m$. Let
$F(t)=\nu\{x\in\mathcal H:\rho(x)\le t\}$.
Using the identity
\[
 (1-\rho)^m=\int_{\rho}^{1}m(1-t)^{m-1}\,dt
\]
and interchanging the finite sum over inputs with the integral, the
average abstention probability is
\[
 m\int_0^1 F(t)(1-t)^{m-1}\,dt
 \le4m\int_0^1 t(1-t)^{m-1}\,dt
 =\frac4{m+1}.
\]
Here \eqref{eq:rho-tail} supplies the inequality.
\end{proof}

\begin{proof}[Proof of \cref{thm:main-randomized}]
First take $\mu_n=U_n$ and choose the anchors using $\nu=U_n$.
Then $\beta=2^{-n}$. Equations~\eqref{eq:geometric} and
\eqref{eq:reciprocal} give
\begin{equation}
 \E_{x\sim U_n,\omega}\tau(x)\le1+4n\ln2.
 \label{eq:uniform-random-samples}
\end{equation}
Each sample and certificate test takes polynomial time, so the total
expected running time is polynomial. \Cref{lem:truncation} proves the
uniform truncation bound.

For a general efficiently samplable $\mu_n$, define
\begin{equation}
 \nu_n=\tfrac12\mu_n+\tfrac12U_n.
 \label{eq:mixture}
\end{equation}
This distribution is efficiently samplable and has full support. It
satisfies
\[
 \nu_n(x)\ge2^{-n-1},\qquad \mu_n(x)\le2\nu_n(x).
\]
Choose the two anchors using the weights $\nu_n$ and use fresh
independent samples from $\nu_n$. Then
\begin{equation}
 \E_{x\sim\mu_n,\omega}\tau(x)
 \le2\E_{x\sim\nu_n,\omega}\tau(x)
 \le2\bigl(1+4(n+1)\ln2\bigr).
 \label{eq:general-random-samples}
\end{equation}
The same domination inequality, applied to the failure probabilities
in \cref{lem:truncation}, gives the bound $8/(m+1)$.

The interpreter contains the fixed algorithms for $f$ and for sampling
$\mu$. The possibly inefficient choices of the anchors are made only
in the existence proof for the advice. There is no restriction that the
sampler for $\mu_n$ use $O(n)$ random bits; its worst-case polynomial
running time is sufficient.
\end{proof}

\begin{remark}[What the expectation does and does not say]
The expectation in \eqref{eq:main-random-mean} is over both the input and
the interpreter's random choices. It is not a polynomial bound for
$\E_{\omega}T_R(x,b_n;\omega)$ on every individual input. In particular,
we do not infer ordinary-advice membership in $\BPP$ or $\mathrm{ZPP}$
from this theorem.
\end{remark}

\subsection{Small-bias translates}
\label{sec:smallbias}

A finite nonempty multiset $S$ in $\F_2^r$ is $\lambda$-biased if
\[
 \left|\E_{s\in S}(-1)^{\langle u,s\rangle}\right|\le\lambda
 \quad\text{for every }u\ne0.
\]
Multisets are used so that a uniform index into an explicit construction
has exactly the stated distribution.

\begin{lemma}[Translate hitting bound]
\label{lem:translate}
If $S$ is $\lambda$-biased and $C\subseteq\F_2^r$ has density
$\rho>0$, then
\begin{equation}
 \Pr_{v\sim U_r}[(v+S)\cap C=\varnothing]
 \le\frac{\lambda^2(1-\rho)}{\rho}.
 \label{eq:translate}
\end{equation}
\end{lemma}
\begin{proof}
Use the uniform inner product on real functions on $\F_2^r$, and define
$Ag(v)=\E_{s\in S}g(v+s)$. The characters
$\chi_u(v)=(-1)^{\langle u,v\rangle}$ form an orthonormal basis and satisfy
$A\chi_u=(\E_s\chi_u(s))\chi_u$. Therefore
$\|Ag\|_2\le\lambda\|g\|_2$ for every mean-zero $g$.
For $g=\ind_C-\rho$, we have $\|g\|_2^2=\rho(1-\rho)$.
At every missed translate, $Ag(v)=-\rho$. If the missed fraction is
$\theta$, then
$\theta\rho^2\le\|Ag\|_2^2\le\lambda^2\rho(1-\rho)$.
\end{proof}

We use an explicit construction of Alon, Goldreich, H\aa stad, and
Peralta \cite{AGHP92}: a bias-$\eta$ multiset on $r$ bits has size
$O(r^2/\eta^2)$, polynomial-time indexed generation after setup bounded by
its size times a polynomial. The precise parameters and the construction
are stated and proved in \cref{lem:smallbias} in
Appendix~\ref{app:technical}.

The setup in \cref{lem:smallbias} is intentionally bounded in terms of
the multiset size, not asserted to be polynomial in $h$. In the
polynomial-time heuristic the multiset is polynomial-sized. In the
exact average-time decider, setup is performed only when its search
stage is reached; its full cost is charged to that stage below.

\subsection{Exact deterministic search with three linear-size fields}
\label{sec:deterministic}

Choose $p,q$ using $\nu=U_n$. For $1\le j\le n$, let $S_j$ be a fixed
explicit $2^{-j/2}$-biased multiset in $\F_2^n$, with
\begin{equation}
 |S_j|\le Cn^2 2^j
 \label{eq:stage-size}
\end{equation}
for an absolute constant $C$. For example, in \cref{lem:smallbias} take
$h_j=\lceil\log n\rceil+\lceil j/2\rceil$.
Fix a shift $v\in\F_2^n$. After the direct tests, search successively
through the certificate candidates in
\[
 v+S_1,\ v+S_2,\ \ldots,\ v+S_n.
\]
Stop upon finding a certificate of either sound type. If all stages
fail, enumerate every string of length $n$. This final search succeeds
by \cref{prop:anchors}, for every input and every shift.

\begin{proof}[Proof of \cref{thm:main-deterministic}]
For an unresolved input $x$, put $\rho=|C(x)|/2^n$. Let $A_j(v,x)$ be
the event that stage $j$ is reached and $A_*(v,x)$ the event that the
exhaustive fallback is reached. Define the \emph{charged stage cost}
\begin{equation}
 B_v(x)=\sum_{j=1}^n |S_j|\ind_{A_j(v,x)}
          +2^n\ind_{A_*(v,x)},
 \label{eq:charged-cost}
\end{equation}
and put $B_v(x)=0$ on directly decided inputs. This charges the entire
size of each reached stage even if a certificate is found early.
Consequently it bounds the number of candidate tests, and
$\poly(n)(1+B_v(x))$ bounds all computation, including the lazy setup
of the finite-field representations.

Reaching stage $j\ge2$ implies missing $v+S_{j-1}$. By
\cref{lem:translate},
\[
 \Pr_v[A_j(v,x)]\le\min\{1,2^{-(j-1)}/\rho\},\qquad
 \Pr_v[A_*(v,x)]\le\min\{1,2^{-n}/\rho\}.
\]
No independence between stages is needed; they use the same shift.
Thus
\begin{align}
 \E_v B_v(x)
 &\le |S_1|+\sum_{j=2}^n |S_j|
       \min\{1,2^{-(j-1)}/\rho\}
       +2^n\min\{1,2^{-n}/\rho\}\notag\\
 &\le 2Cn^2+\frac{2Cn^2(n-1)+1}{\rho}
 =O(n^3/\rho).
 \label{eq:per-input-stage-cost}
\end{align}
Using \cref{lem:reciprocal} with $\beta=2^{-n}$ gives
\begin{equation}
 \E_v\E_{x\sim U_n}B_v(x)=O(n^4).
 \label{eq:mean-stage-cost}
\end{equation}
There is therefore a single shift $v_n$ whose uniform mean charged cost
has this bound. Store $p,q,v_n$ and the empty-class flags: $3n+O(1)$
bits. The interpreter is deterministic and exact on every input; all
randomness above serves only to select good fixed advice.
The maximum charged cost is $O(n^2 2^n)$, so worst-case time is
$2^n\poly(n)$.
\end{proof}

\begin{corollary}[One advice sequence, all polynomial timeouts]
\label{cor:timeouts}
There is a fixed $3n+O(1)$ ordinary advice sequence such that, for every
fixed $c>0$, an errorless polynomial-time truncation of the corresponding
deterministic decider has uniform abstention at most $n^{-c}$.
\end{corollary}
\begin{proof}
If the mean running-time bound is $Q(n)$, use the timeout
$\lceil Q(n)n^c\rceil$ and return $\unk$ on timeout. Markov's inequality
bounds the timeout probability by $n^{-c}$. No advice is changed.
\end{proof}

\subsection{One-step witnesses}

Let $\mu$ be any probability distribution on $X$; full support is not
required in this section. For every $x\in X$, define
\begin{equation}
 W_x=\begin{cases}
   \{y\in P:f(x,y)=x\},&x\in P,\\
   \{y\in N:f(x,y)=y\},&x\in N,
 \end{cases}
 \qquad w(x)=\mu(W_x).
 \label{eq:one-step-witnesses}
\end{equation}
Here equality of the strings is allowed; in particular $x\in W_x$.
A correctly labeled member of $W_x$ certifies the label of $x$ using
one selector evaluation.

\begin{lemma}[One-step witness lower tail]
\label{lem:one-step}
For every $t\ge0$,
\begin{equation}
 \mu\{x:w(x)\le t\}\le4t.
 \label{eq:one-step-tail}
\end{equation}
\end{lemma}

\begin{proof}
Set $A=\{x\in P:w(x)\le t\}$. Then
\begin{align*}
 t\mu(A)
 &\ge\sum_{x\in A}\mu(x)w(x)\\
 &\ge\sum_{x,y\in A}\mu(x)\mu(y)\ind[f(x,y)=x]\\
 &=\frac12\left(\mu(A)^2+\sum_{x\in A}\mu(x)^2\right)
 \ge\frac12\mu(A)^2.
\end{align*}
For distinct $x,y$, commutativity makes exactly one of the two ordered
pairs contribute; diagonal pairs contribute their full weight.
Hence $\mu(A)\le2t$. The identical argument with the selection direction
reversed gives
$\mu\{x\in N:w(x)\le t\}\le2t$. Add the two inequalities.
\end{proof}

\subsection{Encoding all sample labels by one cut}

\begin{lemma}[Hamiltonian-path label compression]
\label{lem:path}
Given a list of at most $M$ strings in $X$, a deterministic algorithm
using $O(M^2)$ selector evaluations constructs an ordering of the
distinct strings
\[
 y_1\to y_2\to\cdots\to y_k.
\]
There is an integer $h\in\{0,\ldots,k\}$ such that
\begin{equation}
 \chi_L(y_i)=\begin{cases}1,&i\le h,\\0,&i>h.
 \end{cases}
 \label{eq:cut-labels}
\end{equation}
Thus the entire membership labeling of the sample is determined by at
most $\lceil\log(M+1)\rceil$ advice bits, once the sample is known.
\end{lemma}

\begin{proof}
Remove duplicates using a fixed deterministic rule. Construct the path
by insertion. To insert a new vertex $z$ into an existing path, put it
immediately before the first vertex it beats, or append it if there is
none. If it is inserted before position $j>1$, every earlier vertex,
including the preceding vertex, beats $z$; the two new adjacent edges
therefore have the required orientation. Insertion at the first
position and appending at the end are justified by the same rule.
There are at most $O(M^2)$ comparisons.

An adjacent pair labeled $0,1$ would have its negative vertex beating
its positive vertex, contradicting selectivity. Every binary sequence
with no adjacent $0,1$ transition is a prefix of ones followed by a
suffix of zeros. This proves \eqref{eq:cut-labels}.
The path need not be a transitive ordering of the tournament; only its
adjacent edges and the cross-class orientation are used.
\end{proof}

\subsection{Short-advice heuristics from one translate}
\label{sec:heuristics}

Assume $\mu_n=G_n(U_{r(n)})$, where $G_n$ and the polynomially bounded
integer $r(n)\ge1$ are uniformly polynomial-time computable. A sampler
using no coins can be supplied one unused coin. For a polynomial-time
computable accuracy $0<\varepsilon\le1/2$ with
$1/\varepsilon=\poly(n)$, we will construct a deterministic errorless
heuristic. The sampler description and the rule choosing the accuracy
are part of the fixed interpreter, not its advice.

\begin{proposition}[An explicit linear-leading-term bound]
\label{prop:heuristic-simple}
Under these hypotheses, average abstention at most $\varepsilon$ is
achievable with ordinary advice of length
\begin{equation}
 r(n)+2\log r(n)+2\log(1/\varepsilon)+O(1).
 \label{eq:heuristic-simple}
\end{equation}
All conclusive answers are correct on every input, including inputs
outside the support of $\mu_n$.
\end{proposition}
\begin{proof}
Apply \cref{lem:smallbias} on $r=r(n)$ bits with bias
$\lambda=\varepsilon/8$, obtaining $M=O(r^2/\varepsilon^2)$ points.
For a shift $v$, use the sample list $(G_n(v+s):s\in S)$.
For each input $x$, the set
\[
 C_x=\{z\in\bits^r:G_n(z)\in W_x\}
\]
has uniform density $w(x)=\mu_n(W_x)$. By \cref{lem:one-step}, inputs
with $w(x)\le\varepsilon/8$ have total measure at most
$\varepsilon/2$. For any remaining input, \cref{lem:translate} bounds
its probability of being missed by
$\lambda^2/w(x)\le\varepsilon/8$. Therefore
\begin{equation}
 \E_v\Pr_{x\sim\mu_n}[\text{sample misses }W_x]
 \le\varepsilon/2+\varepsilon/8<\varepsilon.
 \label{eq:heuristic-simple-error}
\end{equation}
Fix a shift attaining this bound. By \cref{lem:path}, one cut position
in a Hamiltonian path encodes the correct labels of the entire sample.
Store this position together with $v$. The interpreter regenerates the
sample, reconstructs its labels, and returns the label certified by any
sampled one-step witness; otherwise it returns $\unk$.

The advice length is $r+\lceil\log(M+1)\rceil$, proving
\eqref{eq:heuristic-simple}. All inference is sound, and all construction
and path computations are polynomial since $M=\poly(n)$.
\end{proof}

Integrating the same witness-density tail gives the sharper logarithmic
overhead in \cref{thm:main-heuristic}; the additional estimate and its
application are in \cref{lem:integrated-miss,prop:heuristic-refined} of
Appendix~\ref{app:technical}.

\begin{proof}[Proof of \cref{thm:main-heuristic}]
Take $G_n$ to be the identity and $\varepsilon=n^{-c}$ in
\cref{prop:heuristic-simple,prop:heuristic-refined}. Then $r=n$ and
$\log(1+\log(1/\varepsilon))=\log\log n+O_c(1)$.
The final assertion follows from \cref{cor:timeouts}.
\end{proof}

\section{Hidden-core estimates and parameter optimization}\label{app:core}
\label{sec:core}

The next construction addresses correctness on every input rather than
uniform coverage. Its small core is negligible under the uniform input
distribution, but hides independent labels that a worst-case decider must
still determine.

\begin{theorem}[Finite query--advice tradeoff]
\label{thm:finite-lower}
Let $q,k\ge1$, $\ell\ge0$, and $1\le m<n$ be integers with
$4q\le2^{n-m}$. Suppose
\begin{align}
 \ell+3&\le k(n-m-\log(4q))-m,\label{eq:lower-condition-one}\\
 \ell+3&\le\left\lfloor\frac{2^m}{k}\right\rfloor2^{-k}\log e.
 \label{eq:lower-condition-two}
\end{align}
For every deterministic black-box decider with query bound $q$, a random
instance from the distribution below is decided by no advice string of
length at most $\ell$ with probability strictly greater than $1/2$.
\end{theorem}

\subsection{The distribution}

Write $x=(\beta,w)$ with bucket $\beta\in\bits^m$ and tail
$w\in\bits^{n-m}$. Choose independently and uniformly a tail $w_\beta$
and a label $\lambda_\beta\in\bits$ for every bucket, and one orientation
coin $c_{\beta\gamma}$ for every unordered pair of distinct buckets.
Put
\[
 t_\beta=(\beta,w_\beta),\qquad
 \Core=\{t_\beta:\beta\in\bits^m\},\qquad
 P=\{t_\beta:\lambda_\beta=1\}.
\]
The tournament is defined as follows. Outside $\Core$, lexicographically
smaller strings beat larger strings. Every member of $\Core$ beats every
point outside $\Core$. For two core points, opposite labels make the
positive point win; equal labels use their independent orientation coin.
This is a legal instance, regardless of all the choices.

Core membership has density $2^{m-n}$, but the input may itself be a core
point. The lower bound must therefore quantify the cost of discovering
\emph{other} core points and the ambiguity that remains after discovery.
It does not merely assume that uniform sampling is the only available
query strategy.

\subsection{Ambiguous runs}

Fix one advice string. On input $t_\beta$, define $\AR_\beta$ to be the
run of $D$ with the following modified answers: in a comparison between
$t_\beta$ and $t_\gamma$, $\gamma\ne\beta$, make $t_\beta$ lose when
$\lambda_\gamma=1$ and win when $\lambda_\gamma=0$. Answer all other
queries using $\tau$. These modifications define a fixed tournament
oracle, so the query bound still applies.

Let $E_\beta$ be the other buckets whose actual core point occurs as an
endpoint of any query in $\AR_\beta$. Let $D_\beta\subseteq E_\beta$ be
the buckets actually compared with $t_\beta$, and let $o_\beta$ be the
output. These sets concern actual endpoint equalities, whether or not
the procedure can recognize a discovered core point.

The ambiguous run depends on neither $\lambda_\beta$ nor any incident
coin $c_{\beta\gamma}$. Its transcript coincides with the true transcript
whenever, for every $\gamma\in D_\beta$ with
$\lambda_\gamma=\lambda_\beta$, the coin $c_{\beta\gamma}$ points in the
ambiguous direction. When the labels differ, agreement is automatic.
Consequently, after fixing all variables except $\lambda_\beta$ and
its incident coins, the probability that the true run agrees with the
ambiguous run and outputs the wrong label is at least
\begin{equation}
 \tfrac12\,2^{-|D_\beta|}.
 \label{eq:ambiguous-error}
\end{equation}
Repeated queries of an edge do not create additional constraints.

\begin{lemma}[Discovering core points]
\label{lem:discovery}
For each bucket $\beta$,
\begin{equation}
 \Pr[|E_\beta|\ge k]\le(4q)^k2^{-k(n-m)}.
 \label{eq:discovery}
\end{equation}
In particular,
\begin{equation}
 \Pr[\max_\beta|E_\beta|\ge k]
 \le2^{m-k(n-m-\log(4q))}.
 \label{eq:discovery-union}
\end{equation}
\end{lemma}
\begin{proof}
Fix all labels and orientation coins, and fix the input tail $w_\beta$.
Expose the at most $2q$ query endpoints sequentially, augmenting the
observable transcript with whether each endpoint actually equals the
core point in its bucket. This augmentation is used only in the
analysis; it grants no information to the decider.

For an undiscovered bucket $\gamma\ne\beta$, every untried tail gives
the same previous augmented history: none of its points seen so far was
special. Conditional on this history, its hidden tail remains uniform
over at least
\[
 2^{n-m}-2q\ge2^{n-m-1}
\]
possibilities. The conditional probability that an endpoint discovers a
new bucket is therefore at most $2^{1-(n-m)}$. This remains true if a
query's two endpoints are exposed in either fixed order, because both
endpoints were chosen before their answers were given.

If at least $k$ different buckets are discovered, some $k$ endpoint slots
are new discoveries. For each fixed set of slots, iterated conditioning
bounds its probability by $2^{k(1-(n-m))}$. There are at most $(2q)^k$
choices. This proves \eqref{eq:discovery}, uniformly in the fixed labels,
coins, and input tail. Averaging and then taking a union bound over the
$2^m$ possible input buckets proves \eqref{eq:discovery-union}.
\end{proof}

\begin{lemma}[Few discoveries leave fresh label-and-edge choices]
\label{lem:fresh}
For every fixed choice of all tails,
\begin{equation}
 \Pr\bigl[D\text{ is correct on all of }\Core
       \ \text{and}\ \max_\beta|E_\beta|<k\bigr]
 \le(1-2^{-k})^{\lfloor2^m/k\rfloor}.
 \label{eq:fresh}
\end{equation}
\end{lemma}
\begin{proof}
Expose labels and coins in rounds, maintaining a set $U$ of used buckets,
initially empty. The invariant is that every already exposed label
belongs to $U$, and both endpoints of every already exposed coin belong
to $U$. No global success event is conditioned on during this process.

In a round choose the first bucket $\beta\notin U$. Simulate
$\AR_\beta$, exposing the labels of discovered buckets and the coins
among them as needed. Stop unsuccessfully for the event in
\eqref{eq:fresh} if $|E_\beta|\ge k$. The simulation has not read
$\lambda_\beta$ or any coin incident with $\beta$.

If the simulation finishes with fewer than $k$ discoveries, expose
$\lambda_\beta$ and the coins $c_{\beta\gamma}$ for
$\gamma\in D_\beta$. All are still independent fair variables,
conditional on this round's preceding history: $\beta$ was unused and
none of these variables was needed to simulate the ambiguous run.
With conditional probability at least
\[
 \tfrac12\,2^{-|D_\beta|}\ge2^{-k},
\]
the output is wrong and every needed equal-label coin points in the
ambiguous direction. In that event the true run errs, so the event in
\eqref{eq:fresh} fails. One may equivalently demand the ambiguous
direction from every exposed incident coin, a possibly stronger event
with the same lower bound.

Otherwise add $\{\beta\}\cup E_\beta$ to $U$. At most $k$ buckets are
added, and the exposure invariant is preserved. Thus, unless an earlier
round has already excluded the event, there is an unused bucket for each
of $\lfloor2^m/k\rfloor$ rounds. Conditional on surviving earlier rounds,
the probability of surviving the next is at most $1-2^{-k}$; early
termination due to many discoveries only reduces it. Multiplication
proves the bound.
\end{proof}

\begin{proof}[Proof of \cref{thm:finite-lower}]
For one fixed advice string, \cref{lem:discovery,lem:fresh} bound the
probability that it decides the entire instance by
\begin{equation}
 2^{m-k(n-m-\log(4q))}
 +(1-2^{-k})^{\lfloor2^m/k\rfloor}.
 \label{eq:fixed-advice-probability}
\end{equation}
The first term is at most $2^{-\ell-3}$ by
\eqref{eq:lower-condition-one}. The second is at most
\[
 \exp\bigl(-\lfloor2^m/k\rfloor2^{-k}\bigr)
 =2^{-\lfloor2^m/k\rfloor2^{-k}\log e}
 \le2^{-\ell-3}
\]
by \eqref{eq:lower-condition-two}. There are
$\sum_{j=0}^{\ell}2^j=2^{\ell+1}-1$ eligible advice strings. A union
bound makes the probability that any succeeds strictly smaller than
$2^{\ell+1}\cdot2^{-\ell-2}=1/2$.
\end{proof}

\subsection{Quadratic asymptotics and the query tradeoff}

\begin{corollary}[Quadratic advice for polynomially many queries]
\label{cor:quadratic}
There is an absolute constant $C$ such that, whenever
$1\le q\le2^{n-4}$ and
\[
 B(n,q)=\left\lfloor\frac{(n-\log q)^2}{4}-Cn\log(n+2)\right\rfloor
 \ge0,
\]
every deterministic $q$-query black-box decider has a legal instance
that no advice of length at most $B(n,q)$ decides. For
$q=\poly(n)$ this gives $n^2/4-O_q(n\log n)$ necessary advice bits.
For $q\le2^{\varepsilon n}$ with fixed $0<\varepsilon<1$, the bound is
$(1-\varepsilon)^2n^2/4-O(n\log n)$.
\end{corollary}
\begin{proof}
Write $d=n-\log q$ and take
\[
 m=\lfloor d/2\rfloor,\qquad
 h=\lceil3\log n\rceil+2,\qquad k=m-h.
\]
We verify the nonvacuous range $k\ge1$; by increasing the absolute
constant $C$, all excluded parameters have $B(n,q)<0$.
Since $d-m\ge d/2\ge2$, we have $4q\le2^{n-m}$.
Let
\[
 \ell_0=\left\lfloor k(d-m-2)-m-3\right\rfloor.
\]
Then \eqref{eq:lower-condition-one} holds for $\ell=\ell_0$, and
\[
 \ell_0\ge d^2/4-O(n\log n).
\]
For example, expand
$(m-h)(d-m-2)=m(d-m)-2m-h(d-m-2)$ and use
$m(d-m)\ge d^2/4-1$, $h=O(\log n)$, and $d\le n$.
If $\ell_0<0$ there is again no nonnegative lower bound to assert.

For $k\ge1$ the elementary inequality
$\lfloor2^m/k\rfloor\ge2^m/(2k)$ gives
\[
 \left\lfloor\frac{2^m}{k}\right\rfloor2^{-k}\log e
 \ge\frac{2^h}{2k}\log e
 \ge\frac{4n^3}{n}\log e=4n^2\log e,
\]
where $k\le m\le n/2$. Meanwhile $\ell_0+3\le n^2/4$.
Thus \eqref{eq:lower-condition-two} holds, and
\cref{thm:finite-lower} applies. The two displayed special cases follow
by substitution. Any smaller nonnegative advice bound also fails.
\end{proof}

\subsection{Finite randomized tradeoff and its optimization}
\label{sec:randomized-lower}
The complete lifting lemma is \cref{m:lift} in the main text. Here we record
its finite-parameter consequences and check the quadratic asymptotics.

For the hidden-core distribution of \cref{sec:core}, the discovery and
fresh-variable lemmas imply, for every fixed deterministic $Q$-query
procedure, the bound
\begin{equation}
 \Delta(n,m,k,Q)=
 2^{m-k(n-m-\log(4Q))}
 +(1-2^{-k})^{\lfloor2^m/k\rfloor}
 \label{eq:randomized-delta}
\end{equation}
on the probability of correctness on all inputs, provided
$4Q\le2^{n-m}$. In fact they bound the larger event of correctness on
all core points. Crucially, $\Delta$ is independent of the fixed
procedure, its advice, and any hardwired random tape.

\begin{theorem}[Randomized finite query--advice tradeoff]
\label{thm:randomized-finite}
Let $q,k\ge1$, $\ell\ge0$, and $1\le m<n$ be integers, and put
$Q=t_nq$, where $t_n=12(n+2)+1$. Suppose $4Q\le2^{n-m}$ and
\begin{align}
 \ell+4&\le k(n-m-\log(4Q))-m,
 \label{eq:randomized-condition-one}\\
 \ell+4&\le\left\lfloor\frac{2^m}{k}\right\rfloor2^{-k}\log e.
 \label{eq:randomized-condition-two}
\end{align}
For every randomized black-box decider with query bound $q$ and any
finite coin bound, a random instance from the hidden-core distribution
has no bounded-error correct ordinary advice string of length at most
$\ell$, with probability strictly greater than $2/3$.
\end{theorem}
\begin{proof}
The first summand of \eqref{eq:randomized-delta} is at most
$2^{-\ell-4}$ by \eqref{eq:randomized-condition-one}. The second is at
most
\[
 \exp\!\left(-\left\lfloor2^m/k\right\rfloor2^{-k}\right)
 \le2^{-\ell-4}
\]
by \eqref{eq:randomized-condition-two}. Thus
$\Delta\le2^{-\ell-3}$. For each fixed advice, apply
\cref{m:lift} with $\delta=\Delta$. A union bound over
the $2^{\ell+1}-1$ advice strings gives
\begin{align*}
 \Pr_I[\exists a,\ |a|\le\ell:\ a\text{ is bounded-error correct}]
 &\le(2^{\ell+1}-1)\frac43\Delta\\
 &<2^{\ell+1}\frac43\,2^{-\ell-3}=\frac13.
\end{align*}
Only advice strings are counted; the number of random tapes does not
occur in the estimate.
\end{proof}

\begin{corollary}[Quadratic randomized advice, independent of coin length]
\label{cor:randomized-quadratic}
There is an absolute constant $C_{\rm R}$ such that, whenever
$1\le q\le2^{n-4}$ and
\begin{equation}
 B_{\rm R}(n,q)=
 \left\lfloor\frac{(n-\log q)^2}{4}
       -C_{\rm R}n\log(n+2)\right\rfloor\ge0,
 \label{eq:randomized-quadratic}
\end{equation}
every randomized $q$-query black-box decider, with any finite coin bound,
has a legal instance on which no ordinary advice of length at most
$B_{\rm R}(n,q)$ gives error at most $1/3$ on every input.
For polynomial $q$ the lower bound is $n^2/4-O_q(n\log n)$.
For $q\le2^{\varepsilon n}$, where $0<\varepsilon<1$ is fixed, it is
$(1-\varepsilon)^2n^2/4-O(n\log n)$.
\end{corollary}
\begin{proof}
Set $Q=t_nq$, $d'=n-\log Q$, and choose
\[
 m=\lfloor d'/2\rfloor,\qquad
 h=\lceil3\log n\rceil+2,\qquad k=m-h.
\]
In the nonvacuous range take
\[
 \ell_0=\left\lfloor k(d'-m-2)-m-4\right\rfloor.
\]
As in the proof of \cref{cor:quadratic}, the first finite condition
holds by definition, while
\[
 \left\lfloor\frac{2^m}{k}\right\rfloor2^{-k}\log e
 \ge4n^2\log e>\ell_0+4.
\]
The query condition follows from $d'\ge4$. These statements apply when
$k\ge1$ and $\ell_0\ge0$; taking $C_{\rm R}$ sufficiently large makes
\eqref{eq:randomized-quadratic} negative in the excluded parameter range.
For completeness, $d'<4$ or $k<1$ implies $n-\log q=O(\log n)$,
which is absorbed by $C_{\rm R}n\log(n+2)$.
Finally,
\begin{align*}
 \ell_0
 &\ge\frac{(n-\log Q)^2}{4}-O(n\log n)\\
 &\ge\frac{(n-\log q)^2}{4}-O(n\log n),
\end{align*}
because $\log t_n=O(\log n)$ and $0\le n-\log q\le n$.
Increase $C_{\rm R}$ once more to ensure $B_{\rm R}(n,q)\le\ell_0$,
and apply \cref{thm:randomized-finite}. Substitution gives the two
special cases.
\end{proof}

\paragraph{Inverse-polynomial advantage.}
The same leading quadratic bound holds for pointwise success
$1/2+\gamma(n)$ when $1/\gamma(n)=\poly(n)$.
To see this, put $p=1/2-\gamma$. For $0<\gamma<1/2$, optimizing the
exponential-moment bound for $t$ independent repetitions gives
\[
 \Pr[\text{majority error}]
 \le\bigl(2\sqrt{p(1-p)}\bigr)^t
 =(1-4\gamma^2)^{t/2}\le e^{-2\gamma^2t}.
\]
Choose odd $t\ge(n+2)\ln2/(2\gamma^2)$ and repeat the transfer proof.
The amplified query count is $Q=O(qn/\gamma^2)$, so for polynomial
$q$ and inverse-polynomial $\gamma$ the loss remains $O(n\log n)$.
Exact success ($\gamma=1/2$) is already covered by the deterministic
bound. This observation concerns a specified uniform advantage on every
input, not arbitrary unbounded-error acceptance probabilities.

\section{The query-only constant and the single-core family}\label{app:constant}
\label{sec:constant}

For polynomial-time interpreters, the comparison remains the
$1/4$ hidden-core lower coefficient against Ko's coefficient $1$.
A smaller upper coefficient is proved below only in the unrestricted
internal-computation query model of \cref{m:model}.

\subsection{Ko's polynomial-time upper bound}

For comparison, Ko's upper bound has an elementary black-box
implementation using $n^2+n$ bits and at most $n$ queries
\cite{Ko83}. In a nonempty remaining positive subtournament, a
minimum-outdegree vertex is beaten by at least half the other vertices.
Store it and remove it together with all positive vertices that beat it.
The remaining count $s$ becomes at most $(s-1)/2$, so at most $n$
iterations exhaust at most $2^n$ positives. Store the selected vertices
in $n$ fields of $n$ bits and a presence mask of $n$ bits. Accept exactly
when the input equals or beats a present positive anchor.

\subsection{A half-quadratic universal family in the query model}
\label{sec:half}

\paragraph*{Model qualification.} Arbitrary finite oracle-independent
computation is allowed here. The universal family below is computable,
but its indexing is not shown polynomial-time. The theorem therefore does
not establish $\Psel\subseteq\Pclass/(n^2/2+O(n\log n))$ and does not
improve Ko's polynomial-time coefficient.

Call a tuple $S=(v_1,\ldots,v_k)$ \emph{valid} for $(\tau,P)$ if all
its entries lie in $P$, and every $x\in P$ either is an entry or beats
some entry. The empty tuple is valid exactly when $P$ is empty. Given
a valid tuple, the decoder accepts exactly when $x$ equals or beats an
entry. This rejects every negative input.

\begin{lemma}[Many sufficiently good greedy choices]
\label{lem:quantile}
In a tournament on $s\ge1$ vertices, let the outdegrees in increasing
order be $d_1\le\cdots\le d_s$. For $1\le h\le s$,
\begin{equation}
 d_h\le\frac{s+h-2}{2}.
 \label{eq:degree-quantile}
\end{equation}
In particular, for $h=\lceil s/(n+1)\rceil$ there are at least $h$
vertices whose outdegree is at most $\alpha s-1/2$, where
\[
 \alpha=\frac{n+2}{2(n+1)}.
\]
\end{lemma}
\begin{proof}
The $s-h+1$ largest-degree vertices have total outdegree at most
\[
 \binom{s-h+1}{2}+(s-h+1)(h-1).
\]
Their smallest degree is $d_h$. Divide by $s-h+1$ to obtain
\eqref{eq:degree-quantile}. Since $h\le s/(n+1)+1$, the last assertion
follows.
\end{proof}

\begin{lemma}[Uniform tuple success probability]
\label{lem:tuple-prob}
Define
\begin{equation}
 H_n=\binom{n+2}{2}+(n+2)\lceil\log(n+1)\rceil.
 \label{eq:H}
\end{equation}
Choose $K$ uniformly from $\{0,\ldots,n+2\}$ and then choose $K$
independent uniform vertices of $X$. For every legal instance, the
resulting tuple is valid with probability at least
\begin{equation}
 p_n=\frac{2^{-H_n}}{n+3}.
 \label{eq:tuple-prob}
\end{equation}
\end{lemma}

\begin{proof}
The empty-positive case follows by choosing $K=0$. For $P\ne\varnothing$,
consider an auxiliary adaptive process used only in the analysis.
On a nonempty residual positive tournament $R_j$ of size $s_j$, choose
uniformly among its $h_j=\lceil s_j/(n+1)\rceil$ smallest-outdegree
vertices, breaking degree ties canonically. Append the chosen vertex
$v_j$, and remove it and all residual vertices that beat it. The new
residual set is the out-neighborhood of $v_j$, so
\[
 s_{j+1}\le\alpha s_j-1/2\le\alpha s_j.
\]
Continue until no positive vertices remain. Every resulting tuple is valid.
Since
\[
 2^n\alpha^{n+2}
 =\tfrac14\left(1+\frac1{n+1}\right)^{n+2}<1,
\]
each path has length $k\le n+2$. The displayed inequality holds directly
for $n=1$, and for $n\ge2$ follows from
$(1+1/(n+1))^{n+2}<e^{4/3}<4$.

Fix a terminal path of length $k$, with residual sizes
$s_0,\ldots,s_{k-1}$. Its probability under the auxiliary process is
$\prod_jh_j^{-1}$. The probability of its vertex sequence under $k$
independent uniform draws is $N^{-k}$. Their ratio is
\[
 \prod_{j=0}^{k-1}\frac{h_j}{N}.
\]
The reverse shrinkage inequalities and $s_{k-1}\ge1$ give
$s_j\ge\alpha^{-(k-1-j)}$. Hence the ratio is at least
\begin{equation}
 ((n+1)N)^{-k}\alpha^{-k(k-1)/2}=2^{-F(k)},
 \qquad F(k)=k(n+\log(n+1))-\frac b2k(k-1),
 \label{eq:path-ratio}
\end{equation}
where $b=\log(1/\alpha)\in(0,1)$.
For $0\le k\le n+1$,
\[
 F(k+1)-F(k)=n+\log(n+1)-bk\ge\log(n+1)-1\ge0.
\]
Thus $F(k)\le F(n+2)$. Furthermore,
\begin{align*}
 F(n+2)
 &= (n+2)\left(\frac{n-1}{2}+\log(n+1)
       +\frac{n+1}{2}\log\left(1+\frac1{n+1}\right)\right)\\
 &\le(n+2)\left(\frac{n-1}{2}+\log(n+1)+\frac{\log e}{2}\right)
 \le H_n.
\end{align*}
Here $(\log e)/2<1$.

For every terminal path, the independent-tuple probability, including
its length choice, is at least $2^{-H_n}/(n+3)$ times its auxiliary
probability. Sum over all terminal paths; their auxiliary probabilities
sum to one. This proves \eqref{eq:tuple-prob}. The auxiliary choices
may depend on the instance, but the actual sampling distribution does not.
\end{proof}

\begin{theorem}[Half-quadratic advice with linear query complexity]
\label{thm:half}
There is a single deterministic black-box decoder, with unrestricted
finite internal computation, that decides every legal length-$n$
tournament instance using at most $n+2$ comparisons and ordinary advice
of length
\begin{equation}
 b_n=H_n+2n+\lceil\log(n+3)\rceil+2
     =\frac{n^2}{2}+O(n\log n).
 \label{eq:half-advice}
\end{equation}
\end{theorem}

\begin{proof}
There are at most $2^{\binom N2+N}\le2^{N^2}$ legal instances.
Draw $M=2^{b_n}$ independent tuples from the distribution of
\cref{lem:tuple-prob}. A fixed instance receives no valid tuple with
probability at most $e^{-Mp_n}$. The choice of $b_n$ ensures
$Mp_n\ge4N^2$. A union bound over all instances gives total failure
probability at most
\[
 2^{N^2}e^{-4N^2}<1.
\]
Thus there exists an ordered family of $M$ tuples containing a valid
tuple for every legal instance.

Fix, for every $n$, the lexicographically first such family. It depends
only on $n$. A uniform machine can compute it by exhaustive search over
finite families and legal instances; testing this universal property is
a finite computation over internally represented tournaments, not a
query to the actual input oracle. The computation terminates by the
existence argument. Advise the $b_n$-bit index of a valid tuple.
Reconstruct it and test whether the input equals or beats an entry.
There are at most $n+2$ entries, so this uses at most $n+2$ actual
oracle comparisons. Correctness follows from validity.
\end{proof}

Together with \cref{cor:randomized-quadratic}, the query-only leading-constant gap is consequently
\[
 \frac14n^2-O(n\log n)
 \quad\hbox{versus}\quad
 \frac12n^2+O(n\log n).
\]
The lower bound continues to allow bounded-error randomness; the upper
bound is deterministic. Efficient indexing of a covering family with
these parameters would be an additional result, not a consequence of
the counting argument.

\subsection{The single-core family on typical instances}

The following statement specifies the sense in which this distribution
matches the lower bound. It does not bound the advice complexity of
\emph{every} instance in its support.

\begin{proposition}[Typical-instance upper bound]
\label{prop:typical}
Fix an integer $d\ge1$. For all sufficiently large $m$, a deterministic
black-box interpreter decides an instance from \cref{sec:core}, using
\begin{equation}
 (m+d)(n-m)+4m+O(1)
 \label{eq:typical-advice}
\end{equation}
ordinary advice bits and $O(m)$ queries, with probability at least
$1-2^{-d}-\exp(-\Omega(m))$ over the instance. Here $m$ is a specified
function of $n$. For even $n$ and $m=n/2$, the bound is
$n^2/4+O_d(n)$.
\end{proposition}
\begin{proof}
Store the label mask of the first $4m$ buckets and the tails of the first
$m+d$ positively labeled buckets in that mask. The mask identifies
which buckets those tails belong to. With probability
$1-\exp(-\Omega(m))$, the mask contains enough positives: its number
of ones is binomial with mean $2m$, while $m+d\le3m/2$ for large $m$.
This binomial-tail estimate also follows directly by applying the
exponential Markov inequality to the sum of independent fair bits.

The interpreter accepts if the input equals or beats any of these
positive anchors, and rejects otherwise. Every negative is rejected.
Conditional on all labels and tails and on the existence of the anchors,
each unadvised positive loses to all $m+d$ anchors with probability
$2^{-(m+d)}$, because their equal-label edge coins are independent.
A union bound over at most $2^m$ positives bounds the chance of any
false rejection by $2^{-d}$. The bit count and query bound are immediate.
\end{proof}

For any prescribed \emph{fixed} $\eta>0$, choose a fixed $d$ large enough
that $2^{-d}<\eta$. This yields success on a fraction arbitrarily close
to one, with $O_\eta(n)$ lower-order advice overhead. It is not a
probability-$1-o(1)$ assertion with a uniform $O(n)$ constant, and it
does not preclude stronger lower bounds on rare members of the same
family. Multiscale cores are one possible way to seek a stronger leading
constant, not a logically necessary consequence of this proposition.

\paragraph{Why a naive nested-core argument does not add costs.}
Suppose an inner core $C$ lies below all outer positives and above all
outer negatives. Any advised positive in $C$ is already beaten by every
outer positive, and any advised negative in $C$ already beats every
outer negative. Thus inner anchors can certify the outer layers, rather
than requiring independent advice for them. This does not rule out a
stronger multiscale construction, but it invalidates the immediate
addition of single-core lower bounds. No lower coefficient above $1/4$
is established here.

\section{Joint oracle construction and advice models}
\label{sec:oracle}\label{app:oracle}

\subsection{The joint sparse-length construction}

For $n\ge2$, put
\begin{align}
 g(n)&=\max\left\{0,
    \left\lfloor n^2/4-n(\log n)^2\right\rfloor\right\},
    \label{eq:g}\\
 h(n)&=\max\left\{0,
    \left\lfloor n-(\log n)^2\right\rfloor\right\}.
    \label{eq:h}
\end{align}
Set $g(0)=g(1)=h(0)=h(1)=0$. The lower-order losses dominate the
logarithmic losses for every fixed polynomial query bound.

\paragraph{Ordinary randomized advice.}
For an oracle $W$, we write $L\in\BPP^W/g(n)$ if there exist a
probabilistic oracle machine $M$, a polynomial $p$, and strings $a_n$ with
$|a_n|\le g(n)$ such that $M^W(x,a;R)$ halts within $p(n)$ steps for every
$x\in\bits^n$, every advice string $a$ of length at most $g(n)$, and every
coin tape $R$, and
\[
 \Pr_R[M^W(x,a_n;R)=\chi_L(x)]\ge2/3
 \quad\text{for every }n\ge0\text{ and }x\in\bits^n.
\]
Bounded error is required only under the designated advice $a_n$; no
bounded-error gap is imposed on other advice strings. The separation
holds under this more permissive convention, and therefore also under
the convention requiring a bounded-error gap on every input--advice pair.
A polynomial-time errorless oracle heuristic is interpreted with the
same designated-advice convention, with soundness and coverage as in
\cref{m:model}.

\begin{theorem}[A joint average--worst-case oracle separation]
\label{thm:oracle}
There are an oracle $W$ and a language $L\in\PselW$ with all the following
properties.
\begin{enumerate}[(i)]
\item No ordinary-advice bounded-error polynomial-time interpreter decides
$L$ with advice budget $g(n)$:
\begin{equation}
 L\notin\BPP^W/g(n).
 \label{eq:oracle-lower}
\end{equation}
The interpreter may use any polynomial number of coins.
\item For every polynomial-time randomized $W$-oracle errorless heuristic
$H$, there is a length $n$ at which every advice string of length at most
$h(n)$ that is sound for $L$ has uniform coverage strictly less than $1/2$.
Thus no designated advice sequence of that length budget gives soundness
and coverage at least $1/2$ at every length.
\item The same language has all the upper bounds in
\cref{thm:main-deterministic,thm:main-randomized,thm:main-heuristic},
relative to $W$. In particular it has $3n+O(1)$-bit exact deterministic
advice with polynomial uniform mean time and $n+O_c(\log n)$-bit
polynomial-time errorless advice with uniform abstention at most $n^{-c}$.
The randomized mean-time bound allows every exactly polynomial-time
$W$-samplable input distribution. It also has the linear coin-dependent
advice of \cref{prop:dependent-advice}.
\end{enumerate}
\end{theorem}

\begin{proof}
\emph{Requirements and clocks.}
Enumerate all pairs consisting of an oracle advice program and one of two
requirement types: worst-case decision or errorless heuristics. Repeat
every pair at arbitrarily large stage indices. For example, if the pairs
are enumerated as $R_1,R_2,\ldots$, assign $R_j$ to each stage $i$ whose
$2$-adic valuation is $j-1$. At stage $i$, clock the assigned program
$M_i$ at $n^i$ steps on every input, permitted advice, and coin tape;
coin use is consequently at most $n^i$. Normalize timeout and malformed
outputs by a fixed convention for the requirement type. Every actual
polynomial-time decider or heuristic is represented at a sufficiently
large clock exponent, and in each required type.

Choose a fixed sufficiently large power of two $n_1$ and set
\begin{equation}
 n_{i+1}=2^{n_i}.
 \label{eq:sparse-lengths}
\end{equation}
The initial constant can be chosen so that, for every stage $i$,
\begin{align}
 n_i^i&\le2^{n_i-4},
 &g(n_i)&\le B_{\rm R}(n_i,n_i^i),
 \label{eq:stage-inequalities}\\
 \frac{(2n_i^i+1)(2^{h(n_i)+1}-1)}{2^{n_i}}&<\frac12.
 &&\label{eq:heuristic-stage-inequality}
\end{align}
Here $B_{\rm R}$ is from \cref{cor:randomized-quadratic}; the first
inequality explicitly enforces its query-range hypothesis.
To check simultaneous feasibility, put $s_i=\log n_i$. The quadratic
estimate is
\[
 \frac{(n_i-is_i)^2}{4}-C_{\rm R}n_i\log(n_i+2)
 \ge \frac{n_i^2}{4}-(i/2+C_{\rm R}+1)n_i s_i
\]
for large $n_i$, so $s_i\ge i/2+C_{\rm R}+2$ suffices for the second
inequality in \eqref{eq:stage-inequalities}. At the chosen large lengths,
$h(n_i)\le n_i-s_i^2$, and the left side of
\eqref{eq:heuristic-stage-inequality} is at most
\[
 6n_i^i2^{-s_i^2}=2^{-s_i^2+is_i+\log6}.
\]
It is below $1/2$ once $s_i\ge i+5$. The recurrence
$s_{i+1}=2^{s_i}$ maintains these inequalities and $is_i\le n_i-4$
after a sufficiently large initial choice. It also ensures the
nonvacuity of the finite lower bounds. Active lengths are recognizable
in polynomial time by iterating the recurrence with capped binary
arithmetic and stopping before exceeding the given input length.

\emph{Oracle records.}
At each active length choose a legal instance $(\tau_i,P_i)$ and set
$L\cap\bits^{n_i}=P_i$. At all other lengths let $L$ be empty.
Use two disjoint record types:
\begin{align}
 E(u,v)&=00uv,
 &|u|=|v|=n_i,\\
 B(x)&=01x0^{\,2^{n_i}-n_i},
 &|x|=n_i.
 \label{eq:oracle-records}
\end{align}
For $E(u,v)$ require $u<v$ lexicographically; its bit is $1$ exactly
when $u$ beats $v$ in $\tau_i$. The membership record $B(x)$ has length
$2^{n_i}+2$ and bit $\chi_{P_i}(x)$; its length determines $n_i$.
Malformed records and all other strings receive answer zero.

\emph{Stage isolation.}
At stage $i$, all earlier instances have been fixed. A run of $M_i$ on
length $n_i$ can access neither its own membership records, whose length
is $2^{n_i}+2>n_i^i$, nor later records, whose shortest possible length
is $2n_{i+1}+2>n_i^i$. Thus on every tape and for every permitted advice,
the only accessible answers depending on the current instance are its
same-length tournament edges. Replace every other accessible answer by
its already fixed value. This yields a black-box procedure making at
most $q=n_i^i$ edge queries per run; the earlier finite oracle tables can
be absorbed into its unrestricted internal computation. They do not
depend on the current instance.

\emph{Worst-case stages.}
Apply \cref{cor:randomized-quadratic} and
\eqref{eq:stage-inequalities}. Choose $(\tau_i,P_i)$ so that for every
advice $a$ of length at most $g(n_i)$ there is some $x\in\bits^{n_i}$
with
\[
 \Pr_R[M_i^W(x,a;R)\ne\ind[x\in P_i]]>1/3.
\]
This diagonalizes against the probabilistic machine directly, not merely
against a deterministic tape enumeration.

\emph{Heuristic stages.}
Apply \cref{m:freezing} to the current heuristic with query bound
$n_i^i$ and advice bound $h(n_i)$. Choose its transitive tournament
instance. Every sound advice string has coverage bounded by the left
side of \eqref{eq:heuristic-stage-inequality}, hence strictly less than
$1/2$. In particular, no advice gives both zero error and the required
coverage on this slice.

In either type, fix all current edge and membership records permanently.
These are finite choices: the bounded input, advice, and tape spaces
allow exhaustive inspection, although no efficiency is required of the
oracle construction. Later stages change no answer accessible to any
run of the current machine. The repeated enumeration covers every
polynomial-time program with a sufficiently large clock, proving both
lower-bound assertions.

\emph{A global polynomial-time selector.}
If neither input has active length, both are negative, so return either
canonically. If exactly one has active length, return it. If both have
the same active length, return the winner according to the corresponding
edge record, or the common string for equal inputs. If both lengths are
active and $r<s$, the shorter string $x$ has a membership record of length
$2^r+2\le s+2$. Query it, returning $x$ when it is positive and the longer
string otherwise. This is correct whenever either input is positive.
All length tests and the padded query take polynomial time in the
combined input length. Canonical treatment of the arguments makes the
selector commutative. Therefore $L\in\PselW$.

Finally, the upper-bound interpreters evaluate only the selector and,
where specified, the sampler. Their combinatorial proofs relativize to
this same $W$. \Cref{prop:dependent-advice} applies to every relativized
P-selective language, completing the last assertion as well.
\end{proof}

\begin{corollary}[Tight average-case and worst-case orders on one language]
\label{cor:joint-mean}
For the same $W,L$, no exact decider with ordinary advice of length at
most $h(n)$ has polynomial uniform mean running time. This includes
zero-error randomized exact deciders whose mean is over both inputs and
coins. Nevertheless deterministic exact mean-time computation and
polynomial-time errorless heuristics have ordinary $O(n)$ advice for
this language, whereas worst-case bounded-error polynomial-time
computation has advice complexity $\Theta(n^2)$.
\end{corollary}
\begin{proof}
If an exact decider had uniform mean time at most a polynomial $Q(n)$,
truncate it after $\lceil2Q(n)\rceil$ steps and return $\unk$ on timeout.
Markov's inequality gives coverage at least $1/2$, with no errors and
unchanged advice, contradicting \cref{thm:oracle}(ii). This argument also
applies when the mean includes internal randomness. The linear upper
bounds are supplied by \cref{thm:oracle}(iii), while
$h(n)=n-o(n)$ gives the linear lower order. The quadratic lower order
follows from \eqref{eq:oracle-lower} and $g(n)=(1/4-o(1))n^2$; Ko's
polynomial-time upper bound from \cref{sec:constant} relativizes and
supplies the quadratic upper order.
\end{proof}

\paragraph{Scope.}
The time guarantee distinguishes exact mean-time computation from
worst-case polynomial time. The heuristic lower bound uses zero error,
not merely high average accuracy. The finite hard tournaments need not
be computable by unrelativized polynomial-time selectors; the oracle
construction is what turns their slices into one globally selectable
language. The oracle lower bounds exclude advice budgets required to hold
at every length: the diagonal failures occur at active lengths. They do
not assert hardness at every length; the language is empty at inactive
lengths.

\subsection{Nondeterministic advice and opposing relativizations}

\begin{proposition}[Relativized nondeterministic linear advice]
\label{prop:np-advice}
For every oracle $A$,
\begin{equation}
 \mathrm{P}^{A}\text{-}\mathrm{sel}
 \subseteq\NP^{A}/(n+1)\ \cap\ \mathrm{coNP}^{A}/(n+1).
 \label{eq:np-advice}
\end{equation}
The NP verifier can use exactly $n$ witness bits, with a dummy witness
in direct or empty-slice cases.
\end{proposition}
\begin{proof}
This is the usual tournament proof of the optimal nondeterministic
upper bound \cite{HT96}, with the linear witness accounting of
\cite{HNP98}. For the NP interpreter, advise a positive
minimum-outdegree vertex $p$ and one flag recording whether the slice
is empty. If nonempty, accept when $x=p$, when $x\to p$, or when a
guessed $n$-bit vertex $z$ satisfies $x\to z\to p$.
\Cref{cor:unweighted} proves completeness inside the positive
tournament; selectivity proves soundness on negative inputs. All tests
use the polynomial-time $A$-selector. Apply the reversed-tournament
argument to a negative anchor for the coNP bound.
\end{proof}

\begin{corollary}[Opposing relativizations]
\label{cor:opposing}
For each fixed $0<\varepsilon<1$, each $\mathcal C\in\{\Pclass,\BPP\}$,
and each $\mathcal K\in\{\Psel,\Mc\}$, neither the containment
\[
 \mathcal K\subseteq\mathcal C/O(n)
 \quad\text{nor the containment}\quad
 \mathcal K\subseteq\mathcal C/O(n^{2-\varepsilon})
\]
can be established or refuted by an argument that relativizes to every
oracle. In particular, even allowing bounded-error polynomial-time
randomness, a relativizing argument cannot improve the quadratic
worst-case ordinary-advice upper bound by a polynomial factor.
An unrelativized counterexample to any of these containments would
refute hypothesis $\mathsf U$ from \cref{def:unique}, and in particular
imply $\Pclass\ne\NP$. A counterexample with $\mathcal C=\BPP$ would
also imply $\NP\not\subseteq\BPP$.
\end{corollary}
\begin{proof}
The oracle $W$ in \cref{thm:oracle} falsifies all these containments, already
on the P-selective subclass of $\Mc^W$, because
$g(n)$ eventually exceeds every constant multiple of $n^{2-\varepsilon}$
(and of $n$). In the other direction let $A$ be a PSPACE-complete oracle.
Then $\Pclass^A=\NP^A=\PSPACE$: polynomially many polynomial-length
PSPACE queries, and nondeterministic polynomial-space computation, can
be simulated in polynomial space, while the complete oracle supplies
the reverse inclusion. Equation~\eqref{eq:np-advice} then yields
\[
 \mathrm{P}^{A}\text{-}\mathrm{sel}\subseteq\Pclass^A/(n+1).
\]
By \cref{m:mc-rel}, the same oracle gives
$\Mc^A\subseteq\Pclass^A/(3n+5)$. All the containments hold relative to it, since
$\Pclass^A\subseteq\BPP^A$. The same argument with the hypothesis
$\Pclass=\NP$, without an oracle, gives deterministic $(n+1)$-bit advice.
Similarly, if $\NP\subseteq\BPP$, the uniform NP interpreter acting on
$(x,a_n)$ can be simulated by a BPP machine with the same fixed $a_n$.
Together with \eqref{eq:np-advice} and \cref{m:mc-main}, this gives
$\Psel\subseteq\BPP/(n+1)$ and $\Mc\subseteq\BPP/(3n+5)$ and proves the final assertion about
$\NP\not\subseteq\BPP$. Finally, the unrelativized conditional theorem
\cref{m:mc-U,thm:conditional} proves
$\mathsf U\Rightarrow\Mc\subseteq\Pclass/O(n)$, including its P-selective subclass. Under $\mathsf U$
all the containments in the statement therefore hold. Their failure
refutes $\mathsf U$; since $\Pclass=\NP$ implies $\mathsf U$, this
strengthens the earlier $\Pclass\ne\NP$ consequence. The conditional
proof does not use this corollary.
\end{proof}

\subsection{Ordinary advice versus advice chosen after the coins}

For clarity, define an errorless randomness-dependent-advice interpreter
as a polynomial-time machine $M$ with advice $a_{n,r}$ depending on the
random tape $r$, such that for every $x\in\bits^n$,
\[
 M(x,a_{n,r};r)\in\{\chi_L(x),\unk\}\ \text{for every }r,
 \qquad
 \Pr_r[M(x,a_{n,r};r)=\chi_L(x)]\ge2/3.
\]
In the ordinary-advice model the same $a_n$ must work for all $r$.
The following includes a standalone upper-bound argument to make the
comparison independent of an implicit advice-table representation.

\begin{proposition}[Linear coin-dependent advice with logarithmically many coins]
\label{prop:dependent-advice}
For every oracle $A$, every $L\in\mathrm{P}^A\text{-}\mathrm{sel}$ has
an errorless polynomial-time interpreter using $O(n)$
randomness-dependent advice bits and $\log n+O(1)$ coins.
\end{proposition}
\begin{proof}
Fix a length and consider valid pairs of positive and negative anchors,
using flags for empty classes. For every input distribution $\mu$, some
pair makes the one-step direct tests conclusive on at least half the
measure. Indeed, within $P$, the $\mu$-weighted average weight of the
closed in-neighborhood is at least $\mu(P)/2$, since each distinct pair
contributes once and diagonal pairs contribute fully. A positive anchor
therefore certifies at least half the positive mass. The corresponding
closed out-neighborhood argument covers at least half the negative mass.
If a class has zero measure, its contribution is trivially covered.
The two tests are always sound.

Apply finite minimax to the payoff indicating a conclusive answer.
There is a distribution on valid anchor pairs for which every input is
decided with probability at least $1/2$. A triple of independent pairs
has failure probability at most $1/8$ on each input. Each triple needs
$6n+O(1)$ bits and has an errorless polynomial-time interpreter.

Sample $K$ such triples independently, with $K$ a sufficiently large
power of two of order $n$. For a fixed input, the probability that more
than one third of the $K$ triples fail is $\exp(-\Omega(K))$.
For an explicit elementary estimate, if $Z$ counts failures, then
\[
 \Pr[Z\ge K/3]\le2^{-K/3}\E 2^Z
 \le2^{-K/3}(9/8)^K,
\]
which is exponentially small because $(9/8)^3<2$.
Choose $K$ large enough and take a union bound over $2^n$ inputs. Some
fixed list has at most one-third failures on every input.

Use $\log K=\log n+O(1)$ coins to choose a list index and supply the
corresponding triple as advice depending on those coins. The interpreter
never receives the whole list as ordinary advice. Soundness is preserved,
and the conclusive probability is at least $2/3$ for every input.
\end{proof}

\begin{theorem}[Quadratic advice-model separation with polynomially many coins]
\label{thm:advice-models}
For the oracle and language in \cref{thm:oracle}, no bounded-error
polynomial-time $W$-oracle interpreter with ordinary advice of length at
most $g(n)$ decides $L$, even with an arbitrary polynomial bound on its
number of random coins. Nevertheless $L$ has an errorless
coin-dependent interpreter with $O(n)$ advice and $\log n+O(1)$ coins.
In particular,
\begin{equation}
 L\notin\BPP^W/g(n)
 \qquad\text{but}\qquad
 L\in\BPP^W//O(n),
 \label{eq:polynomial-coin-separation}
\end{equation}
where the second inclusion has an errorless version with conclusive
probability at least $2/3$ on every input.
\end{theorem}
\begin{proof}
The ordinary-advice lower bound is the probabilistic diagonalization in
\cref{thm:oracle}, based on \cref{cor:randomized-quadratic}.
A polynomial-time machine automatically uses only polynomially many
coins, and the finite lower bound is independent of the coin bound.
For an errorless polynomial-time interpreter with bounded abstention,
replace $\unk$ by an arbitrary bit; its success probability is still at
least $2/3$, so the same lower bound applies.
The coin-dependent upper bound is \cref{prop:dependent-advice}.
\end{proof}

The improvement over tape enumeration is quantitative and essential.
Enumeration would give $2^rq$ queries for an $r$-coin machine.
\Cref{m:lift} instead costs only $t_nq=O(nq)$ queries,
regardless of $r$. The tapes are used in a distributional averaging
argument, not supplied to the interpreter as additional advice.
For ordinary advice the same $a$ is reused in every repetition. With
coin-dependent advice, the repeated advice strings can themselves vary
with the tapes and with the instance; they cannot be absorbed into one
fixed procedure at zero advice cost. Thus the transfer does not
contradict the linear coin-dependent upper bound.

\section{Full proofs for the conditional containment}
\label{sec:conditional}\label{app:conditional}

The lower bounds leave the selector accessible only through tournament
queries. This section uses its polynomial-time program: an advised
membership verifier can then be compiled into an ordinary Boolean
circuit and passed to a promise-problem algorithm. The resulting
containment is unrelativized and conditional, not a black-box simulation.

\begin{definition}[Deterministic promise-unique satisfiability hypothesis]
\label{def:unique}
Hypothesis $\mathsf U$ states that there is a deterministic polynomial-time
algorithm $U$ on Boolean circuits such that $U(C)=0$ whenever $C$ has
no satisfying assignments and $U(C)=1$ whenever it has exactly one.
On circuits with two or more satisfying assignments, either output is
permitted. The algorithm must halt in polynomial time on every circuit.
\end{definition}

This is the promise problem distinguishing zero from one satisfying
assignment, not the total language of circuits with exactly one satisfying
assignment. It is also not the assertion $\Pclass=\mathrm{UP}$.
The classical isolation theorem implies that $\mathsf U$ entails
$\mathrm{NP}=\mathrm{RP}$ \cite{VV86}; $\Pclass=\mathrm{NP}$ certainly
entails $\mathsf U$. We do not assert a proved strict separation between
$\mathsf U$ and $\Pclass=\mathrm{NP}$. No self-reducibility of the
P-selective language will be assumed.

\begin{theorem}[Conditional deterministic linear advice]
\label{thm:conditional}
If promise Unique-Circuit-SAT has a deterministic polynomial-time
algorithm, then
\begin{equation}
 \Psel\subseteq\Pclass/O(n).
 \label{eq:conditional-main}
\end{equation}
\end{theorem}

The proof will preserve the advice length by exploiting the verifier's
linear witness length, rather than an unrestricted randomized SAT
simulation. In particular, its contrapositive is
\begin{equation}
 \Psel\not\subseteq\Pclass/O(n)\ \Longrightarrow\ \neg\mathsf U,
 \label{eq:counterexample-refutes-U}
\end{equation}
which supplies the strengthened conclusion in \cref{cor:opposing}.

\subsection{A general linear-witness transfer}
\begin{theorem}[Advice-preserving isolation]
\label{thm:generic}
Assume $\mathsf U$. Let $a(n),w(n)$ be polynomially bounded,
polynomial-time computable nonnegative integer functions. Suppose there
are advice strings $a_n$ of length $a(n)$ and a polynomial-time predicate
$R$ such that, for every $x\in\bits^n$,
\begin{equation}
 x\in L \quad\Longleftrightarrow\quad
 \exists z\in\bits^{w(n)}\ R(x,a_n,z)=1.
 \label{eq:advised-NP}
\end{equation}
Then
\begin{equation}
 L\in\Pclass/O(a(n)+w(n)+n).
 \label{eq:generic-conclusion}
\end{equation}
The advice need not be computable, and $L$ need not be recursive.
\end{theorem}

We give the parameter-preserving proof in full. The main points are to
use a linear-length seed for pairwise-independent hashing, to test every
hash prefix, and to amplify with one expander walk instead of independent
linear-length seeds.

\begin{lemma}[Linear-seed isolation with constant success]
\label{lem:isolation}
For circuits on $m\ge1$ input bits, there is, under $\mathsf U$, a
polynomial-time procedure $B(C;s)$ using a seed $s$ of length $3m+1$
with the following properties: if $C$ is unsatisfiable, $B(C;s)=0$ for
all $s$; if $C$ is satisfiable, $\Pr_s[B(C;s)=1]\ge3/16$.
\end{lemma}

\begin{proof}
Choose a uniformly random binary Toeplitz matrix
$T\in\F_2^{(m+1)\times m}$ and an independent
$b\in\F_2^{m+1}$. The matrix has $2m$ diagonals, so
$s=(T,b)$ uses $3m+1$ bits. Set
\[
 h_s(z)=Tz+b.
\]
For distinct $z,z'$, the pair $(h_s(z),h_s(z'))$ is uniform and
independent. Indeed, for every fixed nonzero $d$, the map from the
Toeplitz diagonals to $Td$ has row rank $m+1$: an extreme nonzero
coordinate of $d$ gives distinct leading diagonal variables in the
successive row linear forms. Thus $Td$ is uniform, and the independent
shift $b$ supplies the other uniform coordinate.

For each $k\in\{1,\ldots,m+1\}$ form the circuit
\[
 C_{s,k}(z)=C(z)\wedge\ind[h_s(z)_{1\ldots k}=0^k].
\]
Return $B(C;s)=\bigvee_{k=1}^{m+1}U(C_{s,k})$. Unsatisfiable inputs
are always rejected. If the satisfying set $W$ is nonempty, put
$k=\lceil\log|W|\rceil+1$ and
$Z=|\{z\in W:h_s(z)_{1\ldots k}=0^k\}|$. Then
\[
 \mu=\E Z=|W|2^{-k}\in(1/4,1/2],
 \qquad \E[Z(Z-1)]\le\mu^2.
\]
For nonnegative integer $Z$,
$\ind[Z=1]\ge Z-Z(Z-1)$. Consequently
\begin{equation}
 \Pr[Z=1]\ge\mu-\mu^2\ge3/16.
 \label{eq:isolation-prob}
\end{equation}
On this event $U(C_{s,k})=1$. Testing all $m+1$ prefixes avoids losing
a factor $m$ in the isolation probability. An arbitrary acceptance on
a nonunique restricted circuit is harmless: the original circuit is
then satisfiable. All circuits have size polynomial in $|C|+m$, and
all calls and constructions take polynomial time.
\end{proof}

\begin{lemma}[Short descriptions hitting many dense sets]
\label{lem:walk}
For every fixed $\alpha>0$, and any at most $2^n$ subsets
$G_x\subseteq\bits^r$ each of uniform density at least $\alpha$, there
is a sequence of $O_\alpha(n)$ seeds hitting every $G_x$, with a
description of length $r+O_\alpha(n)+O(1)$. Each seed is reconstructible
from that description in time polynomial in $n+r$.
\end{lemma}

\begin{proof}
Use an explicit constant-degree regular expander family with constant
spectral gap and polynomial-time neighbor computation in the vertex
label length. Such a family, with sizes $2^{r'}$ for a fixed arithmetic
progression of exponents $r'$, follows from the construction of
Reingold, Vadhan, and Wigderson \cite[Theorem~3.3]{RVW02}, taking a
fixed power-of-two base degree. Pad $r$ by $O(1)$ bits to an allowed
$r'$. Lift each $G_x$ by ignoring the padding. Passing to a fixed
constant graph power, if necessary, gives normalized nontrivial
eigenvalues of absolute value at most some $\lambda\le1/2$; the degree
$d$ remains a fixed power of two.

Here is the needed walk estimate. Write $P_G$ for the normalized
adjacency operator with the uniform inner product, let $B$ be the
complement of one lifted good set, and let $\Pi_B$ project to functions
supported on $B$. For such a function $v$,
\begin{align*}
 \langle v,P_Gv\rangle
 &\le\lambda\|v\|_2^2+(1-\lambda)(\E v)^2\\
 &\le[\lambda+(1-\lambda)(1-\alpha)]\|v\|_2^2.
\end{align*}
The smallest eigenvalue is at least $-\lambda$, so
\[
 \|\Pi_BP_G\Pi_B\|\le
 \theta:=1-(1-\lambda)\alpha<1.
\]
A stationary walk of $t$ vertices avoids the good set with probability
\[
 \langle\ind_B,(\Pi_BP_G\Pi_B)^{t-1}\ind_B\rangle
 \le\theta^{t-1}.
\]
Choose $t=O_\alpha(n)$ so that $\theta^{t-1}<2^{-n-2}$.
A union bound over all good sets shows that a positive fraction of walks
hit every one. A walk is described by its initial $r'$-bit vertex and
$t-1$ edge labels of $\log d=O(1)$ bits each. Strong explicitness gives
the reconstruction time claimed. The sets themselves need not be
efficiently recognizable to apply this existence argument.
\end{proof}

\begin{proof}[Proof of \cref{thm:generic}]
Fix $n$ and its designated advice $a_n$. Put $m=w(n)$; if $m=0$ the
predicate can be evaluated directly. Compile $R(x,a_n,z)$, for each
$x$, into a Boolean circuit $C_{x,a_n}$ with exactly $m$ free input bits
$z$. This is a uniform polynomial-time construction, and does not add
free auxiliary input bits to the circuit.

For each positive $x$, define
\[
 G_x=\{s\in\bits^{3m+1}: B(C_{x,a_n};s)=1\}.
\]
By \cref{lem:isolation}, every such set has density at least $3/16$.
There are at most $2^n$ of them. By \cref{lem:walk}, choose one walk
whose seeds hit all of them. The advice consists of $a_n$ and the
walk description; its length is $a(n)+O(m+n)$.

On input $x$, reconstruct the walk, evaluate $B(C_{x,a_n};s)$ at each
visited seed, and return their OR. Every negative input is rejected for
every seed. Every positive input is accepted at some visited seed.
The number of visited seeds is $O(n)$, each produces $m+1$ calls to
$U$, and each constructed circuit has polynomial size. Hence the
interpreter is deterministic polynomial-time. The walk description is
ordinary advice, fixed simultaneously for all inputs at length $n$.
\end{proof}

\subsection{Applying the transfer to P-selective sets}

The proof of \cref{prop:np-advice} gives the classical verifier with
$n+1$ advice bits and exactly $n$ witness bits \cite{HNP98}. One bit flags
an empty positive slice. Otherwise the advice names a positive
minimum-outdegree anchor $p$, and the witness predicate is
\begin{equation}
 R_f(x,p,z)=[f(x,z)=x]\wedge[f(z,p)=z],\qquad z\in\bits^n.
 \label{eq:king-verifier}
\end{equation}
The two-step property in \cref{cor:unweighted} proves completeness;
equality and direct edges are covered by $z=x$ or $z=p$. Soundness
follows because $f(z,p)=z$ forces $z$ positive, and then $f(x,z)=x$
forces $x$ positive. The free inputs of the compiled circuit are exactly
these $n$ witness bits; circuit gates do not introduce additional
witness variables.

\begin{proof}[Proof of \cref{thm:conditional}]
Use the verifier just described, with $a(n)=n+1$ and $w(n)=n$, in
\cref{thm:generic}. The advice bound is $O(n)$.
\end{proof}

\begin{corollary}[Binary-comparator consequence]\label{app:mc-U}
Under $\mathsf U$, $\Mc\subseteq\Pclass/O(n)$, also for lengthwise binary
comparability.
\end{corollary}
\begin{proof}
Apply \cref{thm:generic} to \cref{m:mc-main} with the fixed target bit $b=1$,
advice length $3n+5$, and witness length $5n+12$. In majority mode the verifier
may ignore a padded witness; in certificate mode it uses its fixed-length
encoding. Thus the existential characterization has exactly the witness length
required by the transfer. The same argument with $b=0$ applies to the complement.
All compilation uses the actual polynomial-time comparator, not merely oracle
access to its values.
\end{proof}

\paragraph{Why the hypothesis and model distinctions matter.}
The ability to solve promise-unique instances is used only after the
actual polynomial-time selector has been compiled into an ordinary
Boolean circuit. A black-box tournament oracle supplies no such
circuit description. Likewise, a relativized selector produces an
oracle-dependent verifier, and the unrelativized hypothesis does not
supply a promise-unique solver for circuits containing arbitrary oracle
gates. There is therefore no conflict with the oracle separation.
The conditional implication can be relativized only if the hypothesis
is strengthened to the corresponding oracle-circuit promise solver.

The argument also identifies why $\Pclass=\mathrm{UP}$ cannot simply
replace $\mathsf U$: the family of restricted circuits is not globally
unambiguous; some off-promise circuits have many witnesses. The assumed
algorithm must handle zero- and one-solution instances uniformly while
being total on the remaining circuits. Similarly, applying a generic
$\mathrm{NP}=\mathrm{RP}$ simulation does not by itself control the
coin count by $O(n)$, which is the advice-critical parameter here.

\paragraph{Relation to earlier sufficient conditions.}
A different, algebraic sufficient condition is already known: if all
P-selective sets admit associative polynomial-time selectors, then they
have deterministic linear advice \cite{HHN04}. The present conditional
argument assumes no associative selector and no self-reducibility. Its
specific contribution is the linear-seed isolation and short-walk
accounting applied to the classical linear-witness characterization.
This is a parameter-preserving use of classical ingredients, not a
claim that no earlier sufficient condition for linear advice was known.

\section{Small-bias construction and refined heuristic overhead}
\label{app:technical}

This appendix records the explicit construction used in \cref{sec:upper}
and the additional integration that sharpens its logarithmic advice term.
The main $3n+O(1)$ and $n+O_c(\log n)$ statements do not require the
refined coefficient to express their leading orders.

\subsection{An explicit small-bias construction}

\begin{lemma}[An explicit small-bias multiset {\cite{AGHP92}}]
\label{lem:smallbias}
For $r\ge1$ and $0<\eta\le1$, choose
$h=\max\{1,\lceil\log(r/\eta)\rceil\}$. There is a deterministic
$\eta$-biased multiset in $\F_2^r$ of size $2^{2h}=O(r^2/\eta^2)$.
After a deterministic setup taking at most $2^{2h}\poly(r+h)$ time,
each indexed point is generated in $\poly(r+h)$ time.
\end{lemma}
\begin{proof}
Use a polynomial-basis representation of $F=\F_{2^h}$. For each
$(a,b)\in F^2$, output the vector
\[
 s(a,b)_j=\langle\operatorname{bin}(a^j),\operatorname{bin}(b)\rangle,
 \qquad 0\le j<r.
\]
For nonzero $u\in\F_2^r$, the polynomial $p_u(t)=\sum_{j<r}u_jt^j$
is nonzero and has at most $r-1$ roots. Averaging the character first
over $b$ gives zero unless $p_u(a)=0$, when it gives one. Its bias is
at most $(r-1)/2^h\le\eta$.

For completeness, find a degree-$h$ irreducible polynomial by enumerating
monic candidates and trial-dividing by all monic polynomials of degrees
at most $\lfloor h/2\rfloor$. A reducible candidate has such a divisor.
The total work is at most $2^{3h/2+O(1)}\poly(h)$, within the stated
setup bound. Field arithmetic then generates the points as claimed.
\end{proof}

\subsection{Refining the logarithmic overhead}

A second integration of the same density profile halves the coefficient
of the accuracy term, up to a double logarithm.

\begin{lemma}[Integrated miss bound]
\label{lem:integrated-miss}
If a positive random variable $w\le1$ satisfies
$\Pr[w\le t]\le4t$ for every $t\ge0$, then, for $0<a\le1$,
\begin{equation}
 \E\min\{1,a/w\}\le a+4a\ln(1/a)
 \le4a(1+\ln(1/a)).
 \label{eq:integrated-miss}
\end{equation}
\end{lemma}
\begin{proof}
For each $w\in(0,1]$,
\[
 \min\{1,a/w\}=a+a\int_a^1 t^{-2}\ind[w\le t]\,dt.
\]
Take expectations and use the lower-tail bound inside the integral.
\end{proof}

\begin{proposition}[Refined logarithmic overhead]
\label{prop:heuristic-refined}
In \cref{prop:heuristic-simple}, the advice length can be reduced to
\begin{equation}
 r(n)+2\log r(n)+\log(1/\varepsilon)
   +\log(1+\log(1/\varepsilon))+O(1).
 \label{eq:heuristic-refined}
\end{equation}
\end{proposition}
\begin{proof}
Put $b=1+\lceil\log(1/\varepsilon)\rceil$ and
$a=\varepsilon/(64b)$. Use a multiset of bias at most $\sqrt a$.
By \cref{lem:translate,lem:integrated-miss} and the fact that $w(x)>0$
for $\mu_n$-almost every $x$, its expected average miss probability is
at most $4a(1+\ln(1/a))$. This is at most $\varepsilon$: indeed
\[
 1+\ln(1/a)=1+\ln64+\ln b+\ln(1/\varepsilon)<16b
\]
for $b\ge2$, using $\ln b\le b$ and
$\ln(1/\varepsilon)\le b\ln2$.

The multiset size is $O(r^2b/\varepsilon)$. The shift and cut therefore
have combined length \eqref{eq:heuristic-refined}. The remainder of the
interpreter and its soundness proof are unchanged. The field degree can
be chosen by the integer inequality $2^{2h}\ge64r^2b/\varepsilon$, so
there is no need to manipulate an irrational bias exactly.
\end{proof}

\section{Binary-comparator encodings and scope}
\label{app:mc-encoding}

The complete graph argument and soundness/completeness proof appear in
\cref{m:comparators}. This appendix records the verifier's parsing conventions
and consequences without requiring the reader to reconstruct an exponential
true-literal graph.

\subsection{Advice and certificate layouts}

At length $n$, the advice is
\[
 a_n=\operatorname{code}_2(j)\,s_1s_2s_3,
 \qquad j\in\{0,1,2,3\},\quad |s_i|=n+1.
\]
Each slot names an $n$-bit string and one membership bit. In certificate mode
$j\le2$, the first $j$ slots are the unforced covering literals and unused slots
are ignored. In majority mode $j=3$, the slots name three distinct unforced
covering literals. The $j=0$ header indicates no advised premise, not an empty
language slice. Advice always has exactly $3n+5$ bits. Neither independence,
coverage, nor absence of a short forcing walk is tested by the interpreter;
these facts select the designated advice.

A certificate-mode witness has length exactly $5n+12$: a seven-bit header and
five signed-literal slots of length $n+1$. The header has a type bit, an
advised-slot-index bit, a three-bit forcing length $r\in\{1,2,3,4\}$, and a
two-bit final length $s\in\{0,1,2\}$. Used invalid fields are rejected.
For a forced premise the first slot holds $q$, the next three hold the possible
internal literals of the walk from $\overline q$ to $q$, and the fifth holds the
possible internal literal of the final walk from $q$ to the input literal
$t=\langle x,b\rangle$. The endpoints are not encoded a second time. For an
advised premise the index selects one of the first $j$ advice slots; only $s$
and, when $s=2$, the fifth witness slot are used. Reject this type when $j=0$.
Unused fields and slots are ignored. A final walk of length zero is accepted
only when its endpoints coincide. Every edge involving identical underlying
strings is rejected without making a comparator call.

In majority mode the verifier accepts the empty witness, uses the stored label
on an advised input, and otherwise checks whether the three-row majority equals
$b$. A version with exactly $5n+12$ witness bits may instead ignore that padded
witness in majority mode. This version has the same existential acceptance
condition and is the one used in the advice-preserving isolation corollary.
At $n=0$ the one-vertex true-literal graph has a cover containing its single
literal; certificate mode advises it and uses a zero-edge final walk. Thus no
special additive correction is needed at the empty string.

At most six edge tests are performed in forced-premise mode, at most two in
advised-premise mode, and at most three comparator calls in majority mode. With
comparator time $T_g(n)$, the verifier runs in $O(T_g(n)+n)$ time. Correctness is
required for the designated advice sequence. Malformed advice is rejected, but
well-formed nondesignated advice need not be sound. In particular, the theorem
does not assert that the verifier on all pairs $(x,a)$ defines a language in
$\NP\cap\mathrm{coNP}$; it supplies two advice interpreters with one shared
advice sequence and the explicit quantifiers in \cref{m:eq-mc-verifier}.

\subsection{Lengthwise scope and the inherited lower order}

Every comparison and certificate uses distinct strings at the current input
length. The theorem thus holds for a comparator guaranteed correct only on
those pairs. For an oracle $A$, replacing each comparator call by its
polynomial-time $A$-oracle implementation proves \cref{m:mc-rel} with exactly
the same advice and witness accounting. There is no additional membership
oracle for the language, except for information that may independently be
available in the stipulated $A$.

Since every P-selective language is binary membership-comparable, the
nondeterministic-advice lower bounds of \cite{HT96,HNP98} imply
\[
 \Mc\not\subseteq\NP/n,\qquad
 \Mc\not\subseteq\mathrm{coNP}/n.
\]
These are separate noncontainments; no assertion of one simultaneous witness
language is needed. They show that linear order is necessary for the entire
class, but do not establish optimality of coefficient $3$, additive constant
$5$, the $5n+12$ witness bound, or the six-call bound.

\subsection{A disjunction example and the binary restriction}

Consider a comparator returning $00$ on distinct inputs and $01$ on identical
inputs. It is valid for the full language and every co-singleton language.
Its distinct-input clauses are positive disjunctions. On a slice with no
exceptional negative, the true-literal graph has no edges: when three inputs
exist, three row predictions all give $1$; smaller slices use certificate mode.
On a slice containing an exceptional negative $z$, its true negative literal
implies every other input's positive literal. Advising that negative literal
therefore gives the required short certificates. Neither case requires choosing
one side of a disjunction as an unverified premise.

For a comparator of arity $k\ge3$, an excluded vector yields a $k$-literal
clause, and forcing one literal can require $k-1$ premises. In particular,
nonadjacency does not automatically give the single edge
$\overline h\to q$ used in \cref{m:eq-mc-force}. The binary proof therefore
establishes neither a higher-arity fixed-linear-advice theorem nor a general
extension of the selector-specific errorless mean-time bounds.

\section{Disclosure of AI assistance}
\label{app:ai}

OpenAI ChatGPT was used throughout the development and preparation of this
manuscript, including substantive mathematical assistance rather than only
language editing. The uses included proposing candidate improvements to the
advice bounds, drafting and expanding proofs, checking parameter inequalities,
exploring failure cases, locating and summarizing potentially relevant
literature, and producing LaTeX and submission-format revisions.

The human-led discussion supplied the hidden-core distribution, its proposed
quadratic lower bound, the small-bias improvements, and a sequence of research
and exposition priorities. AI assistance was then used in the development of
the weighted certificate arguments and average-case algorithms. Human feedback  identified the
query-range condition in the oracle construction, the designated-advice
convention, the attribution of the randomized transfer to Yao's principle,
and the need to distinguish unrestricted query computation from
polynomial-time decoding.

The binary-membership-comparability theorem was integrated from a research
draft supplied by the human author. In this integration, ChatGPT condensed and
retypeset the signed-literal graph, independent-cover, short-forcing, and
exact-majority proof; preserved the common advice, witness, and comparator-call
accounting; derived the stated corollaries from the existing transfer; and
reorganized the introduction and results table to distinguish comparator-wide
conclusions from selector-specific ones. 

For this submission version, ChatGPT assisted with reorganizing the merged
manuscript, condensing the main exposition, moving full arguments and explicit
constructions to appendices, converting the references to BibTeX, checking
cross-references and page boundaries, and compiling and inspecting the PDF.

\end{document}